\documentclass[a4paper,twocolumn,11pt,unpublished]{quantumarticle}
\pdfoutput=1
\usepackage[utf8]{inputenc}
\usepackage[english]{babel}
\usepackage[T1]{fontenc}
\usepackage{amsmath}
\usepackage{hyperref}
\usepackage{amsfonts, amssymb}
\usepackage{tikz}
\usepackage{lipsum}
\usepackage{mathrsfs}

\usepackage{graphicx}
\usepackage{dcolumn}
\usepackage{bm}
\hypersetup{
    colorlinks=true,
    citecolor=blue,
    linkcolor=black,
    urlcolor=blue
}
\usepackage{cite}

\usepackage{amsthm}
\newtheorem{thm}{Theorem}
\newtheorem{lem}{Lemma}
\newtheorem{prop}{Proposition}
\newtheorem{cor}{Corollary}

\newtheorem{rem}{Remark}

\begin{document}

\title{Facets of Non-locality and Advantage in Entanglement-Assisted Classical Communication Tasks}

\author{Sumit Rout}
\affiliation{International Centre for Theory of Quantum Technologies (ICTQT), University of Gda{\'n}sk, Jana Ba{\.z}ynskiego 8, 80-309 Gda{\'n}sk, Poland}
\author{Anubhav Chaturvedi}
\affiliation{Faculty of Applied Physics and Mathematics, Gda{\'n}sk University of Technology, Gabriela Narutowicza 11/12, 80-233 Gda{\'n}sk, Poland.}
\author{Some Sankar Bhattacharya}
\affiliation{Física Teòrica: Informació i Fenòmens Quàntics, Universitat Autònoma de Barcelona, 08193 Bellaterra, Spain}
\author{Pawe{\l} Horodecki}
\affiliation{International Centre for Theory of Quantum Technologies (ICTQT), University of Gda{\'n}sk, Jana Ba{\.z}ynskiego 8, 80-309 Gda{\'n}sk, Poland}

\maketitle
\begin{abstract}
We reveal key connections between non-locality and advantage in correlation-assisted classical communication. First, using the \emph{wire-cutting} technique, we provide a Bell inequality tailored to \emph{any} correlation-assisted bounded classical communication task. The violation of this inequality by a quantum correlation is equivalent to its quantum-assisted advantage in the corresponding communication task. Next, we introduce \emph{wire-reading}, which leverages the readability of classical messages to demonstrate advantageous assistance of non-local correlations in setups where no such advantage can be otherwise observed. Building on this, we introduce families of classical communication tasks in a Bob-without-input prepare-and-measure scenario, where non-local correlation enhances bounded classical communication while shared randomness assistance yields strictly suboptimal payoff. 
For the first family of tasks, assistance from any non-local facet leads to optimal payoff, while each task in the second family is tailored to a non-local facet. We reveal quantum advantage in these tasks, including qutrit over qubit entanglement advantage.
\end{abstract}

\section{Introduction}
Shared entanglement and local measurements give rise to \emph{non-local correlations} \cite{Bell1964, Brunner2014review}, which defy classical local-causal explanations. While entanglement does not increase the capacity of classical communication channels, it provides advantages in classical communication tasks \cite{RAC1,RAC2,PhysRevA.60.2737,Deba1,Brukner2002,BRUKNER2003,Brukner2004,Buhrman001,Cleve1997,Cubitt2010,Cubitt2011,Leung2012,Ho2022,PhysRevA.81.042326,Tavakoli2017,Tavakoli2020doesviolationofbell,Tavakoli2021,Yadavalli2022contextualityin,Pauwels22,Pauwels2022,Perry2015,agarwal2025nonlocalityassistedenhancementerrorfreecommunication,Alimuddin2023,Frenkel22}. Specifically, shared entanglement can reduce the classical communication required for some distributed computation \cite{PhysRevA.60.2737,Buhrman001,Cleve1997,Cubitt2010,Cubitt2011,Leung2012,Ho2022,Brukner2002,BRUKNER2003,Brukner2004, Perry2015} (measured in bits) or enhance its probability of success when the communication overhead is bounded \cite{RAC1,RAC2,Frenkel22,Pauwels2022,Alimuddin2023}. This article explores the relationship between quantum non-locality and its implications for communication complexity-like tasks.\\

Non-local correlations are necessary for an advantage in entanglement-assisted one-way noisy classical communication tasks. In this article, using a general proof technique called {\it wire-cutting}, we construct a Bell inequality tailored to any such classical communication task. We demonstrate that the advantage of entanglement assistance in a classical communication task directly corresponds to a proportionate violation of the associated Bell inequality and vice versa. Consequently, the violation of the associated Bell inequality offers a tool to probe for the advantage of quantum non-locality assistance in the classical communication task. Next, we consider the inverse problem regarding the utility of quantum non-locality in classical communication tasks. For a classical channel, we note that the message can be revealed without disturbing the communication. Thus, it can be considered as an observable in the classical communication tasks, a feature we term \emph{wire-reading}.\\

Wire-reading is a powerful tool for witnessing the utility of entanglement assistance to classical communication. Specifically and, perhaps surprisingly, we show that incorporating the classical message into the analysis uncovers quantum advantage in classical communication scenarios where no such advantage could be observed otherwise. In other words, even when the set of all correlations in the initial task could be simulated using shared randomness assistance to the classical channel, reading the classical wire introduces correlations beyond those classically simulable. Equipped with wire-reading, we consider one-way communication tasks in the simplest scenario where the receiver has no inputs. In such minimal one-way communication scenarios, assistance from shared randomness leads to no advantage of quantum communication over a classical channel \cite{Frenkel15}. Consequently, any quantum advantage in this scenario must stem from shared correlations without classical explanation. Specifically, any advantageous entanglement-assisted classical communication protocol must be adaptive in such tasks \cite{Pauwels2022}, {\it i.e.} Bob's measurement choice must depend on the received message.\\

We introduce different families of classical one-way communication tasks in the Bob-without-input scenario, where shared randomness assistance is always sub-optimal. We provide a tight upper bound on the payoff when arbitrary shared randomness assistance to classical communication is allowed for each task. For the first family of tasks, we show that correlations from the non-local facets of the no-signalling polytope perfectly assist the classical channel, {\it i.e.} they maximise the payoff. As parameters associated with the task increase, the advantage drops sharply, making it progressively challenging to demonstrate the utility of shared entanglement. Specifically, some non-local correlations on the isotropic line from a non-local facet might not be useful for the task. However, we can observe advantages of some of these correlations through a different family of tasks. We introduce a class of tasks corresponding to each non-local facet in the no-signalling polytope. We show that assistance from correlations belonging to this specific non-local facet is optimal for these tasks.  As an explicit example, we show the resourcefulness of all non-local correlations along the isotropic line in the no-signalling polytope - connecting a specific $n$-input and $2$-output bipartite non-local extremal/ facet correlation to white noise as long as the noise fraction is strictly less than $0.5$.\\ 

We explicitly provide several instances of quantum advantage for some of the tasks. For a particular task $\mathbb{CS}[d=2,k=3]$, we observe that the maximum payoff achievable using quantum correlations is obtained from a two-qutrit entangled state, while assistance from two-qubit entangled states leads to a lower payoff.\\ 

In the next section, we discuss some mathematical preliminaries, the Bell inequality associated with an entanglement-assisted classical communication task, and wire-reading.

\section{Bell Inequalities from Classical Communication Tasks}
\label{sec:CS to Bell}

We will briefly discuss no-signalling correlations first. For this purpose, consider a non-local scenario involving Alice and Bob, who are space-like separated. They can access some shared correlations $P=\{P(a,b|x,y)\}_{a\in [A],b\in [B],x\in [X],y\in [Y]}$ where $x\in [X]$, $y\in [Y]$ denote the inputs of Alice and Bob, respectively, and $a\in [A]$ and $b\in [B]$ denote their respective outputs. The polytope $\mathcal{NS}$ denotes the convex hull of all no-signalling correlations they share. For any $P=\{P(a,b|x,y)\}_{a,b,x,y}\in\mathcal{NS}$, $P(a|x,y)=P(a|x,y')$ and $P(b|x,y)=P(b|x',y) ~\forall~a\in [A], b\in [B],~ x,x'\in [X], ~y,y'\in [Y]$. We will sometimes shortly regard a no-signalling correlation $P\in \mathcal{NS}$ equivalently as a function $P: [X]\times [Y]\to [A]\times [B]$ with two inputs and two outputs. Local polytope $\mathcal{L}\subset \mathcal{NS}$ denotes the set of local correlations, i.e., the ones obtained from shared randomness. For $P=\{P(a,b|x,y)\}_{a,b,x,y}\in\mathcal{L}$, $P(a,b|x,y)=\sum_{\lambda\in \Lambda} \mu(\lambda)p(a|x,\lambda)p(b|y,\lambda)$ where $\{\mu(\lambda)\}_{\lambda\in \Lambda}$ is a probability distribution over ontic state space $\Lambda$. The extremal points of the Local polytope $\mathcal{L}$ are denoted as $\{P_{L}^{(i)}\}$, and the remaining extremal points of the $\mathcal{NS}$ polytope $\{P_{NL}^{*(i)}\}$ represent the non-local extremal points. A non-local facet of the $\mathcal{NS}$ polytope is a convex hull of some non-local extremal correlations, such that all correlations $P$ on it are either one of these non-local extremal correlations or can only be expressed as a convex mixture of these extremal points. In other words, if a correlation $P$ lies on some son-local facet then $P\neq pP'+(1-p)P_{L}^{(i)}$ where $0\leq p<1$ and $P'\in \mathcal{NS}$. We denote by $P_{WN}$ the correlation in this scenario such that $P_{WN}(a,b|xy)=\frac{1}{AB}~\forall x,y,a,b$. A no-signalling correlation $\Tilde{P}$ lies on the isotropic line connecting the correlation $P$ and $P_{WN}$ if it can be expressed as some convex mixture $pP+(1-p)P_{WN}$ of these two correlations where $0\leq p\leq1$. A non-trivial linear Bell functional can be expressed as  
\begin{equation}\label{eqn:Bellineq}
\mathbf{B}(P)=\sum_{a,b,x,y} c_{x,y}^{a,b}~~P(ab|xy)
\end{equation}
where $c_{x,y}^{a,b}\in \mathbb{R}~ \forall x\in [X],y\in [Y],a\in [A],b\in [B]$ and $P\in\mathcal{NS}$. The Bell functional specifies a hyperplane in the no-signalling polytope such that the local correlations reside on one-half of the hyperplane defined by $\mathbf{B}(P) \leq \beta_L$. Here $\beta_L=\max_{P\in\mathcal{L}} \mathbf{B}(P)$ is the local bound of the Bell inequality. Any violation of Bell inequality, i.e. $\mathbf{B}(P)>\beta_L$, implies that the no-signalling correlation $P$ is non-local, {\it i.e.}, $P\in \mathcal{NS}\backslash\mathcal{L}$. A no-signalling correlation $P$ has quantum realisation if there is a quantum bipartite shared state $\rho_{AB}\in D(\mathbb{H}_A\otimes \mathbb{H}_B)$ and POVMs $\{E_a^x\in \mathcal{B}_+(\mathbb{H}_A):\sum_{a\in[A]} E_{a}^x=\mathbb{I} \} ~\forall x\in [X]$ for Alice and POVMs $\{E_b^y\in \mathcal{B}_+(\mathbb{H}_B):\sum_{b\in[B]} E_{b}^y=\mathbb{I} \} ~\forall y\in [Y]$ for Bob such that $P(a,b|x,y)=Tr(\rho_{AB}~E_a^x\otimes E_b^y)$. $\mathbb{H}_A$ and $\mathbb{H}_B$ are the Hilbert space associated with Alice's and Bob's subsystems, respectively. Although the no-signalling correlations cannot be used for communication, they can be used to assist classical communication in the PM scenario.\\

We briefly discuss one-way classical channels, which are accessible to the parties for the communication task we will introduce later. We will denote a discrete and memoryless classical channel with inputs in the finite set $T$ and output in the finite set $T$ as $\mathscr{T}: T\to T$. For such a channel, $\mathscr{T}(\tau'|\tau)$ denotes the conditional probability of receiving alphabet $\tau'\in T$ when the sent alphabet is $\tau\in T$. In general tasks, the channel could be noisy; however, for some of the tasks, we would consider that Alice and Bob have access to a noiseless channel.  The classical channel is noiseless if $\mathscr{T}(\tau'|\tau)=\delta_{\tau',\tau}$.\\

Now we describe the setup and characterise the communication tasks, which we will consider in this work. We consider a bipartite prepare and measure (PM) scenario involving Alice and Bob. Alice receives input $m\in M$, and Bob can produce output $n\in N$ where $M$ and $N$ are finite sets. Alice can access a one-way classical channel of bounded capacity $\mathscr{T}: T \to T$ (noisy or noiseless). Additionally, they have access to arbitrary shared randomness, which can assist classical communication. However, sharing entanglement or other non-classical correlations is regarded as costly. Using the shared correlation assistance to classical channel $\mathscr{T}$ as well as local pre- and post-processing, Alice and Bob can produce conditional probabilities $\{\mathcal{N}(n|m)\}_{n,m}$. In this PM setup, different tasks can be defined based on the observed conditional probabilities $\{\mathcal{N}(n|m)\}_{n,m}$ of random variables involving the input $m\in M$ of Alice and output $n\in N$ of Bob. The success metric $S$ for a communication task that we will consider in this work is given by some linear function of conditional probabilities $\{\mathcal{N}(n|m)\}_{n,m}$. It can be specified as 
\begin{equation}\label{eq: generic payoff}
S(\mathcal{N})=\sum_{m,n} w^m_{n} ~ \mathcal{N}(n|m)
\end{equation}

Here, $w^m_{n} \in \mathbb{R}~\forall~n\in N, m\in M$. We would denote a communication task in this PM scenario as $\mathbb{C}_{M,N}[\mathscr{T},\{w^m_n\}]$. For the task, the parties can have access to shared correlation (shared randomness, shared entanglement or post-quantum correlations) that is no-signalling.\\ 

For the task $\mathbb{C}_{M,N}[\mathscr{T},\{w^m_n\}]$, consider the scenario when Alice and Bob have access to a no-signalling correlation $P\in \mathcal{NS}$. The distribution $\{\mathcal{N}(n|m)\}_{n,m}$ is realisable using no-signalling correlation assistance along with the classical channel $\mathscr{T}: T\to T$  if there is some $P\in \mathcal{NS}$ and suitable strategies such that such that
\begin{equation} \label{eq: NSRealizableChannel}
 \begin{split}
     \mathcal{N}(n|m) = \sum_{\substack{a, b, x, y,\\ \tau, \tau'}}& 
     \big[ p(x|m) \, p_{e}(\tau|a,m) \, P(a,b|x,y)\times \\
      &~\times\mathscr{T}(\tau'|\tau) \, p(y|\tau') \, p_{d}(n|\tau',b) \big]
 \end{split}
\end{equation}
Here, $\{p(x|m)\}_{x\in[X],m\in M}$ and $\{p(y|\tau')\}_{\tau'\in T, y\in [Y]}$ denotes coding respectively for Alice and Bob for choosing their input of no-signalling resource based on their input/received message. $\{p_{e}(\tau|a,m)\}_{\tau\in T, a\in [A], m\in M}$ is Alice's encoding for the communication based on her input and response from resource correlation. $\{p_{d}(n|\tau',b)\}_{\tau'\in T, n\in N, b\in [B]}$ is Bob's decoding based on his received message and response from resource correlation. Note that, the equation \eqref{eq: NSRealizableChannel} does not makes sense in general if the correlation $P=\{ P(a,b|x,y)\}_{a,b,x,y}\in \mathcal{NS}$ is signalling from Bob to Alice as $\{\mathcal{N}(n|m)\}_{n,m}$ may not be a valid conditional probability distribution (see appendix \ref{appendix: Invalid for signalling}). 

\begin{rem}\label{rem: pre-post processing ns}
   Consider a non-local scenario where Alice's and Bob's input is $m\in M$ and $\tau'\in T$. Their output are respectively $\tau\in T$ and $n\in N$. Suppose Alice and Bob share a no-signalling correlation $P=\{ P(a,b|x,y)\}_{a,b,x,y}\in \mathcal{NS}$. They can apply local pre-processing: Alice sample $x\in[X]$ according to distribution $\{p(x|m)\}_{x\in[X],m\in M}$ based on her input variables $m\in M$, and  Bob samples $y\in [Y]$ according to distribution $\{p(y|\tau')\}_{\tau'\in T, y\in [Y]}$ based on his input variables $\tau'\in T$. Additionally, Alice and Bob can post-process their respective outcomes $a\in[A]$ and $b\in [B]$ based on their local variables $m\in M$ and $\tau'\in T$ to and finally output $\tau\in [T]$ and $n\in N$ respectively. They do this according to distribution $\{p_{e}(\tau|a,m)\}_{\tau\in T, a\in [A], m\in M}$ and $\{p_{d}(n|\tau',b)\}_{\tau'\in T, n\in N, b\in [B]}$, respectively. Thus, the no-signalling correlation they obtain after pre- and post-processing \cite{Wolfe2020quantifyingbell, Schmid2020typeindependent} is $\{P(\tau,n|m,\tau')\}_{\tau,n,m,\tau'}$, which can be expressed as
\begin{align}\label{eq: prepost processing}
P(\tau,n|m,\tau')=\sum_{\substack{a, b, x, y}}& \big[
      p(x|m) p_{e}(\tau|a,m) P(a,b|x,y) \nonumber\\
      &~~\times\,p(y|\tau')  p_{d}(n|\tau',b) \big]
   \end{align}
\end{rem}

We can alternatively express the equation \eqref{eq: NSRealizableChannel} using the transformed no-signalling correlation $\{P(\tau,n|m,\tau')\}_{\tau,n,m,\tau'}$. Consider the distributions  that appears in \eqref{eq: NSRealizableChannel}: $\{p(x|m)\}_{x\in[X],m\in M}$ and $\{p_{e}(\tau|a,m)\}_{\tau\in T, a\in [A], m\in M}$ for Alice, and  $\{p(y|\tau')\}_{\tau'\in T, y\in [Y]}$ and  $\{p_{d}(n|\tau',b)\}_{\tau'\in T, n\in N, b\in [B]}$ for Bob. When the channel $\mathscr{T}$ is absent or maximally noisy, the pre- and post-processing maps alone induce a transformation of the no-signalling correlation $\{P(a,b|x,y)\}_{a,b,x,y}$ into no-signalling correlation $\{P(\tau,n|m,\tau')\}_{\tau,n,m,\tau'}$ (remark \ref{rem: pre-post processing ns}). In contrast, when $\mathscr{T}$ is non-trivial, it couples Alice’s post-processed output $\tau$ to Bob’s preprocessing input $\tau'$ according to probability $\mathscr{T}(\tau'|\tau)$, thereby implementing a \emph{wiring} of the effective correlation. This is evident when the equation \ref{eq: NSRealizableChannel} is rewritten as follows after a simple manipulation using the chain rule for conditional probabilities, along with the assumption that the variables involved are structured according to a directed acyclic graph (DAG):

\begin{align}\label{eq: NSRealizableChannel2}
&\mathcal{N}(n|m)=\sum_{\tau,\tau'}\mathscr{T}(\tau'|\tau) P(\tau,n|m,\tau')
\end{align}

$P(\tau,n|m,\tau')=\sum_{\substack{a, b, x, y}} p(x|m) p_{e}(\tau|a,m) \times$ $P(a,b|x,y)p(y|\tau')  p_{d}(n|\tau',b)$ is the same as in \eqref{eq: prepost processing}. Effectively, Alice on receiving input $m\in M$, uses it as the input to the transformed no-signalling correlation $\{P(\tau,n|m,\tau')\}_{\tau\in T,n\in N,m\in M,\tau'\in T}$. Alice sends the output $\tau\in T$ she obtains as the message using the classical channel $\mathscr{T}$. Bob, upon receiving message $\tau'\in T$, uses it as the input to the shared correlation, and outputs the outcome $n\in N$ obtained from this correlation. Note that \eqref{eq: NSRealizableChannel2} makes sense since the correlation $\{P(\tau,n|m,\tau')\}_{\tau,n,m,\tau'}$ is no-signalling. In general, equation \eqref{eq: NSRealizableChannel2} may not yield valid conditional probability distribution $\{\mathcal{N}(n|m)\}_{m,n}$ if the correlation $\{P(\tau,n|m,\tau')\}_{\tau,n,m,\tau'}$ is Bob to Alice signalling.\\

For the task $\mathbb{C}_{M,N}[\mathscr{T},\{w^m_n\}]$, when Alice and Bob have access to shared randomness, they can jointly randomise over deterministic encodings and decoding based on ontic state $\lambda\in \Lambda$, which follows some distribution $\{p(\lambda)\}_{\lambda \in \Lambda}$. In this case, each deterministic encoding and decoding for Alice and Bob, respectively, can be represented by functions $\mathbb{E}: M\to T$ and $\mathbb{D}: T\to N$. The conditional probability $\{\mathcal{N}_{\Lambda}(n|m)\}_{n,m}$ can be obtained by the parties while using shared randomness assistance to the classical channel $\mathscr{T}: T\to T$ described by a probability distribution $\mathscr{T}(\tau'|\tau)$ if it there is a probability distribution $\{p(\lambda)\}_{\lambda \in \Lambda}$, encoding scheme $\{p_{e}(\tau|m,\lambda)\}_{\tau\in T,\lambda\in\Lambda,m\in M}$ and decoding scheme $\{p_{d}(n|\tau',\lambda)\}_{n\in 
 N,\tau'\in T,\lambda\in \Lambda}$ such that 
\begin{equation} \label{eq: SRRealizableChannel}
 \begin{split}
     \mathcal{N}_{\Lambda}(n|m)=\sum_{\substack{\lambda,\tau, \tau'}} \big[p(\lambda) p_{e}(\tau|m,\lambda)  \mathscr{T}(\tau'|\tau) 
      p_{d}(n|\tau',\lambda)\big]
\end{split}
\end{equation}
The maximum payoff, {\it i.e.}, local bound, for the communication task $\mathbb{C}_{M,N}[\mathscr{T},\{w^m_n\}]$ which can be obtained while using shared randomness assistance to the classical channel $\mathscr{T}$ is given by
\begin{equation} \label{eq: SRbound_generictask}
 \begin{split}
 s_{\Lambda}=\max_{\mathcal{N}_{\Lambda}}S(\mathcal{N}_{\Lambda})=\max_{\substack{\{p_{e}(\tau|\lambda,m)\},\\\{p_{d}(n|\tau',\lambda)\},\\\{p(\lambda)\}}} ~\sum_{m,n}w^m_n\mathcal{N}_{\Lambda}(n|m)
\end{split}
\end{equation}
As the payoff is linear in the conditional probabilities, the local bound while using shared randomness assistance can also be obtained by following some particular deterministic encoding and decoding strategy. If the observed conditional probability distribution $\{\mathcal{N}(n|m)\}_{n,m}$ does not admit a decomposition as expressed in \eqref{eq: SRRealizableChannel}, then it cannot be simulated using arbitrary shared randomness assistance to $\mathscr{T}: T\to T$. $\{\mathcal{N}(n|m)\}_{n,m}$ can be obtained using some non-local correlation assistance to $\mathscr{T}: T\to T$  if there is some $P\in \mathcal{NS}\backslash \mathcal{L}$ such that it can be expressed as in \eqref{eq: NSRealizableChannel} for some encoding and decoding strategy of Alice and Bob. In the task $\mathbb{C}_{M,N}[\mathscr{T},\{w^m_n\}]$, if the observed conditional probability $\{\mathcal{N}(n|m)\}_{n,m}$ can be obtained using some non-local correlation $P$ and leads to a payoff $S(\mathcal{N})=\sum_{m,n}w^m_n \mathcal{N}(n|m)>s_{\Lambda}$ then it essentially shows the advantage of assistance from non-classical correlation to the communication channel.\\

We will now introduce a main ingredient, the \emph{cutting the classical wire} procedure, to obtain an associated Bell scenario and a Bell inequality corresponding to the communication task.\\

\subsection{Cutting the classical wire}\label{subsec: wirecutting}
Corresponding to a communication task $\mathbb{C}_{M,N}[\mathscr{T},\{w^m_n\}]$ with success metric as defined in \eqref{eq: generic payoff}, we consider a bipartite non-local scenario wherein Alice receives an input $m\in M$ and produces an output $\tau\in T$, whereas Bob gets an input $\tau'\in T$ and outputs $n\in N$ (see Figure \ref{fig: PMtoBell}).  
In this non-local scenario, let us consider the following Bell functional tailored to the communication task,

\begin{equation}\label{eq: BellExpression}
    \mathbf{B}_{S,\mathscr{T}}(P)=\sum_{\substack{m,\tau, \tau', n}}w^m_n\mathscr{T}(\tau'|\tau) \, \, P(\tau,n|m,\tau')
\end{equation}
\begin{figure}[h!]
\includegraphics[scale=0.05]{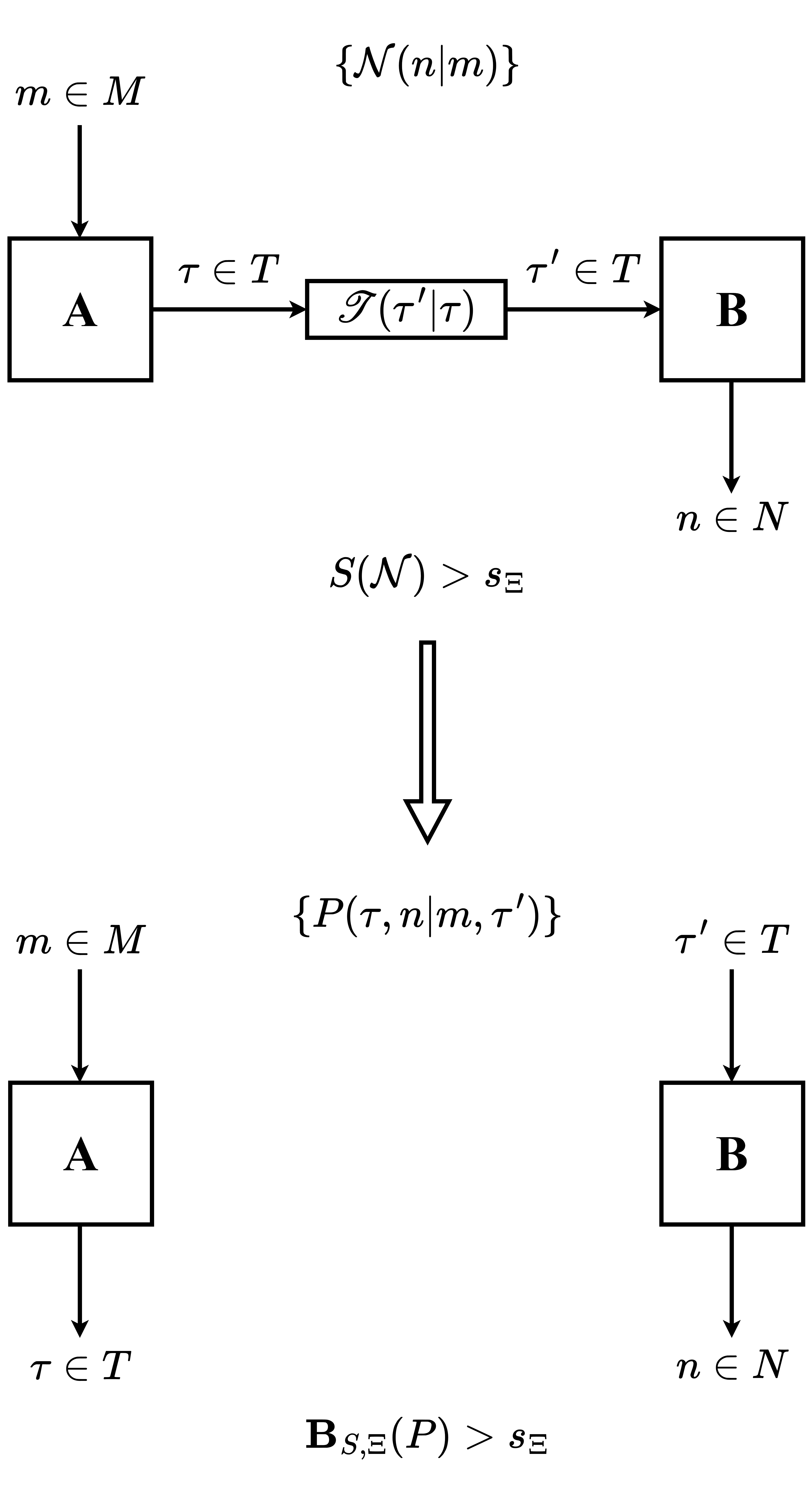}
    \caption{Schematic representation of the proof technique used in Theorem \ref{theo: PMtoBell}, illustrating the process of \emph{cutting the classical wire} to obtain the associated Bell inequality. The first figure (top) represents a task $\mathbb{C}_{M,N}[\mathscr{T},\{w^m_n\}]$ with payoff function $S$ and below is the associated Bell scenario corresponding to Bell functional  $\mathbf{B}_{S,\mathscr{T}}$.}
    \label{fig: PMtoBell}
\end{figure}
Here, $P=\{P(\tau,n|m,\tau')\}_{\tau\in T,n\in N,m\in M,\tau'\in T}$ is some no-signalling correlation in this non-local scenario. Let the local bound for $\mathbb{C}_{M,N}[\mathscr{T},\{w^m_n\}]$ be as given in \eqref{eq: SRbound_generictask}, and observed conditional probability $\{\mathcal{N}_{\Lambda}(n|m)\}_{n,m}$ can be expressed as in \eqref{eq: SRRealizableChannel}. Now, we demonstrate that the inequality $S(\mathcal{N}_{\Lambda})\leq s_{\Lambda}$ in the PM scenario implies the following Bell inequality based on the Bell expression in \eqref{eq: BellExpression}.

\begin{lem} \label{Lemma1}
Given the local bound $s_{\Lambda}$ on the success metric $S$ \eqref{eq: generic payoff}, for the task $\mathbb{C}_{M,N}[\mathscr{T},\{w^m_n\}]$, the following Bell inequality holds for all no-signalling correlations $P=\{P(\tau,n|m,\tau')\}_{m\in M,n\in N, \tau \in T,\tau'\in T}$ possessing a local hidden variable explanation.
\begin{equation} \label{BellInequality}
        \mathbf{B}_{S,\mathscr{T}}(P)\leq s_{\Lambda},
\end{equation}
where the Bell functional $\mathbf{B}_{S,\mathscr{T}}(P)$ is defined in \eqref{eq: BellExpression}.
\end{lem}

\begin{proof}
We use the fact that all observed conditional probability $\{\mathcal{N}_{\Lambda}(n|m)\}_{n,m}$ for the task $\mathbb{C}_{M,N}[\mathscr{T},\{w^m_n\}]$ which can be expressed in the form \eqref{eq: SRRealizableChannel} satisfy the inequality $S(\mathcal{N}_{\Lambda}) \leq s_{\Lambda}$, to demonstrate that all correlations $P=\{P(\tau,n|m,\tau')\}_{m\in M,n\in N, \tau \in T,\tau'\in T}\in \mathcal{L}$, i.e., possessing a \emph{local-hidden variable} explanation of the form,

\begin{equation} \label{LHVExplaination}    P(\tau,n|m,\tau')=\sum_{\lambda}p(\lambda)p_A(\tau|m,\lambda)p_B(n|\tau',\lambda)
\end{equation}

for some probability distributions $\{p(\lambda)\}_{\lambda\in\Lambda}$, $\{p_A(\tau|m,\lambda)\}_{m\in M,\tau\in T, \lambda \in \Lambda}$, and $\{p_B(n|\tau',\lambda)\}_{n\in N,\tau'\in T, \lambda\in \Lambda}$ satisfy the Bell inequality,
$\mathbf{B}_{S,\mathscr{T}}(P)\leq s_{\Lambda}$.\\

Let us assume the contrary, {\it i.e.}, there exist correlations $P$ possessing a \emph{local-hidden variable} model of the form \eqref{LHVExplaination} such that $\mathbf{B}_{S,\mathscr{T}}(P)> s_{\Lambda}$. Then using the resource communication channel $\mathscr{T}$, encoding scheme $p_e(\tau|m,\lambda)=p_A(\tau|m,\lambda)$, the decoding scheme $p_d(n|\tau',\lambda)=p_B(n|\tau',\lambda)$, and shared randomness $\lambda\in\Lambda$ distributed according to the probability distribution $p(\lambda)$ the parties can obtain some conditional probability distribution $\{\mathcal{N}(n|m)\}_{n,m}$. Clearly, $\mathcal{N}(n|m)$ in this case can be expressed as in \eqref{eq: SRRealizableChannel}. Substituting this expression for conditional probability  $\mathcal{N}(n|m)$, we obtain $S(\mathcal{N})=\sum_{m,n} w^m_{n} \mathcal{N}(n|m)=\mathbf{B}_{S,\mathscr{T}}(P)> s_{\Lambda}$, which contradicts our initial assumption $S(\mathcal{N}_{\Lambda})\leq s_{\Lambda}$ for all $\mathcal{N}_\Lambda$ assuming decomposition as in \eqref{eq: SRRealizableChannel}.
\end{proof}

From a geometric perspective, Lemma \ref{Lemma1} implies that the hyperplane $S(\mathcal{N})\leq s_{\Lambda}$ in the space of all observed conditional probability distribution $\{\mathcal{N}(n|m)\}_{n,m}$ implies the existence of a corresponding hyperplane $\mathbf{B}_{S,\mathscr{T}}(P)\leq s_{\Lambda}$ in the no-signalling polytope $\mathcal{NS}$ of Bell correlations  $P: M \times T \to T\times N$ given by conditional probability $P=\{P(\tau,n|m,\tau')\}_{m\in M,n\in N, \tau \in T,\tau'\in T}$. We now turn to the connection between the violation of the inequality $S(\mathcal{N})\leq s_{\Lambda}$ in the task $\mathbb{C}_{M,N}[\mathscr{T},\{w^m_n\}]$ using no-signalling correlation assistance in task $\mathbb{C}_{M,N}[\mathscr{T},\{w^m_n\}]$ and the violation of the corresponding Bell inequality.

\begin{thm} \label{theo: PMtoBell}
If a no-signalling correlation 
$P = \{P(\tau,n|m,\tau')\}_{\tau \in T,\, n \in N,\, m \in M,\, \tau' \in T}$ violates the Bell inequality as defined in \eqref{BellInequality}, {\it i.e.} $\mathbf{B}_{S,\mathscr{T}}(P) > s_{\Lambda} $, then it's assistance in the task $\mathbb{C}_{M,N}[\mathscr{T},\{w^m_n\}]$ leads to the violation of local bound $s_{\Lambda}$ for the payoff, defined in \eqref{eq: generic payoff}. Also corresponding to the violation of the local bound  $s_{\Lambda}$ on the payoff for the task $\mathbb{C}_{M,N}[\mathscr{T},\{w^m_n\}]$ using no-signalling correlation assistance, {\it i.e.} $S(\mathcal{N})>s_{\Lambda}$, there exists a no-signalling correlation $P=\{P(\tau,n|m,\tau')\}_{\tau\in T,n\in N,m\in M,\tau'\in T}$, obtained by cutting the classical wire, which violates the corresponding Bell inequality \eqref{BellInequality}, i.e., $   \mathbf{B}_{S,\mathscr{T}}(P)>s_{\Lambda}$.
\end{thm}

\begin{proof}
Say, the no-signalling correlation $P = \{P(\tau,n|m,\tau')\}_{\tau \in T, n \in N, m \in M, \tau' \in T}$ violates the Bell inequality given in \eqref{BellInequality}, {\it i.e.} $\mathbf{B}_{S,\mathscr{T}}(P) > s_{\Lambda}$. In the task $\mathbb{C}_{M,N}[\mathscr{T},\{w^m_n\}]$, let Alice chooses her input $m\in M$ as her to query to the correlation $P$ and on obtaining output $\tau\in T$ from it, she sends $\tau$ as message through the channel $\mathscr{T}$. On receiving message $\tau'\in T$, Bob uses it as his query to the correlation $P$. He uses the outcome $n\in N$ from the correlation as his output for the task. Thus, the observed distribution $\{\mathcal{N}(n|m)\}_{n,m}$ using this strategy is given as $\mathcal{N}(n|m)=\sum_{\tau,\tau'}\mathscr{T}(\tau'|\tau) P(\tau,n|m,\tau')$. The payoff for the task $\mathbb{C}_{M,N}[\mathscr{T},\{w^m_n\}]$, using \eqref{eq: BellExpression}, is 
\begin{equation}\label{eq: commtobellviolatio}
    S(\mathcal{N})=\sum_{m,n} w^m_{n} \mathcal{N}(n|m)=\mathbf{B}_{S,\mathscr{T}}(P) > s_{\Lambda}
\end{equation}
Now, say there is a conditional probability distribution $\{\mathcal{N}(n|m)\}_{n,m}$, such that $S(\mathcal{N})>s_{\Lambda}$, which can be generated during the task $\mathbb{C}_{M,N}[\mathscr{T},\{w^m_n\}]$ given a no-signalling correlation $\tilde{P}$ assistance to the classical channel $\mathscr{T}: T \to T$ specified by the conditional probability distribution $\mathscr{T}(\tau'|\tau)$. Then $\mathcal{N}(n|m)$ can be expressed as in \eqref{eq: NSRealizableChannel} and \eqref{eq: NSRealizableChannel2}.  The \emph{cutting the classical wire} procedure (see Fig. \ref{fig: PMtoBell}) implies the existence a no-signalling correlation $P=\{P(\tau,n|m,\tau')\}_{\tau\in T,n\in N,m\in M, \tau'\in T}$ which when supplied as assistance to channel $\mathscr{T}$ realizes the distribution $\{\mathcal{N}(n|m)\}_{n,m}$.\\

This can be directly seen from the equations \eqref{eq: NSRealizableChannel} and \eqref{eq: NSRealizableChannel2}.  Using \eqref{eq: NSRealizableChannel}, there is an encoding and decoding strategy for Alice and Bob for realising the distribution $\{\mathcal{N}(n|m)\}_{n,m}$ using the no-signalling correlation $\tilde{P}$ assistance to channel $\mathscr{T}$. Since \eqref{eq: NSRealizableChannel} can be rewritten as the equation \eqref{eq: NSRealizableChannel2}, according to which $\mathcal{N}(n|m)=\sum_{\tau,\tau'}\mathscr{T}(\tau'|\tau) P(\tau,n|m,\tau')$ where $\{P(\tau,n|m,\tau')\}_{\tau,n,m,\tau'}$ is another no-signalling correlation (related to $\tilde{P}$ using remark \ref{rem: pre-post processing ns}). Consequently, the fact that $S(\mathcal{N})>s_{\Lambda}$ implies,
\begin{align}\label{eq: belltocommviolatio}
    \mathbf{B}_{S,\mathscr{T}}(P)&=\sum_{\substack{m ,\tau', \tau, n}}w^m_n \, \mathscr{T}(\tau'|\tau) \, P(\tau,n|m,\tau')\nonumber\\
    &=\sum_{m, n}w^m_n \mathcal{N}(n|m)=S(\mathcal{N}) >s_{\Lambda}
\end{align}
\end{proof}

 A direct consequence of Theorem \ref{theo: PMtoBell} is that, for any given communication task $\mathbb{C}_{M,N}[\mathscr{T},\{w^m_n\}]$, the maximum value for the payoff in \eqref{eq: generic payoff} while using shared quantum correlations is equal to the maximum quantum violation of the associated Bell inequality in \eqref{BellInequality} (see equation \eqref{eq: commtobellviolatio},\eqref{eq: belltocommviolatio}). Say, the maximum payoff using assistance from quantum correlation for the task $\mathbb{C}_{M,N}[\mathscr{T},\{w^m_n\}]$ is $s_Q$. Say, it is achieved by some quantum correlation $P$ assistance. From \eqref{eq: belltocommviolatio}, $\mathbf{B}_{S,\mathscr{T}}(P)=s_Q$. If the highest Bell violation while using quantum correlation is $\mathbf{B}_{S,\mathscr{T}}(P')=\tilde{s}_{Q}\geq s_Q$, which is achieved by $P'$ then the assistance from this correlation in task $\mathbb{C}_{M,N}[\mathscr{T},\{w^m_n\}]$ yields a payoff $\tilde{s}_{Q}$ using \eqref{eq: commtobellviolatio}. Thus, $\tilde{s}_{Q}=s_Q$. This equivalence allows one to probe the maximal achievable quantum advantage in communication task $\mathbb{C}_{M,N}[\mathscr{T},\{w^m_n\}]$ by well-known techniques of find the maximum quantum violation of the associated Bell inequality using non-commuting polynomial optimisation (NPO) \cite{Navas_2008, NVMethod, Navascues2007, Pauwels2022}. Next, we will introduce \emph{wire-reading}, which will be used in the correlation-assisted communication tasks that we will construct.

\subsection{Wire reading}\label{subsec: wirereading}
In the PM scenario, the transmitted message through a classical channel $\mathscr{T}: T\to T$ can be observed without interference. Thus, the communicated alphabet can be recorded without interrupting the communication \cite{Cavalcanti2022, VanHimbeeck2019quantumviolationsin, PhysRevLett.125.050404}. We refer to this as \emph{wire-reading} (see Figure \ref{fig: Wirereading}). As the classical message can be read, we can define a linear payoff for our tasks, which depends on the classical alphabet communicated in addition to Bob's output. Similar to the communication task defined before, let $M$ and $N$ denote the set of inputs and outputs, respectively, for Alice and Bob. Also, they have access to a one-way classical communication channel $\mathscr{T}:T\to T$ where $\mathscr{T}(\tau'|\tau)$ denotes the probability of Bob receiving message $\tau'\in T$ when Alice sends $\tau \in T$. The figure of merit depends on the conditional probability $\{\mathcal{N}(\tau',n|m)\}_{n\in N, m\in M,\tau'\in T}$ of Bob receiving message $\tau'\in T$ and giving output $n\in N$, provided Alice receives $m\in M$ as input, {\it i.e.}, 
\begin{equation}\label{eq: generic wirereding payoff}
{S_{W}}(\mathcal{N})=\sum_{m,\tau',n} w^m_{\tau',n} ~ \mathcal{N}(\tau',n|m)
\end{equation}
The coefficients of the linear payoff function in this case are $w^{m}_{n,\tau'} \in \mathbb{R}$ $\forall n\in N, m\in M,\tau'\in T$. We will denote this task in the PM scenario as $\mathbb{CW}_{M,N}[\mathscr{T},\{w^m_{\tau',n}\}]$. Remarkably, when the classical channel $\mathscr{T}$ is noiseless, this task can also be interpreted using another communication task $\mathbb{C}_{\tilde{M},\tilde{N}}[\mathscr{T},\{\tilde{w}^m_{n}\}]$ in the same scenario whose linear payoff is independent of the transmitted message. More specifically, we have the following:\\

\begin{figure}[h!]
\includegraphics[scale=0.095]{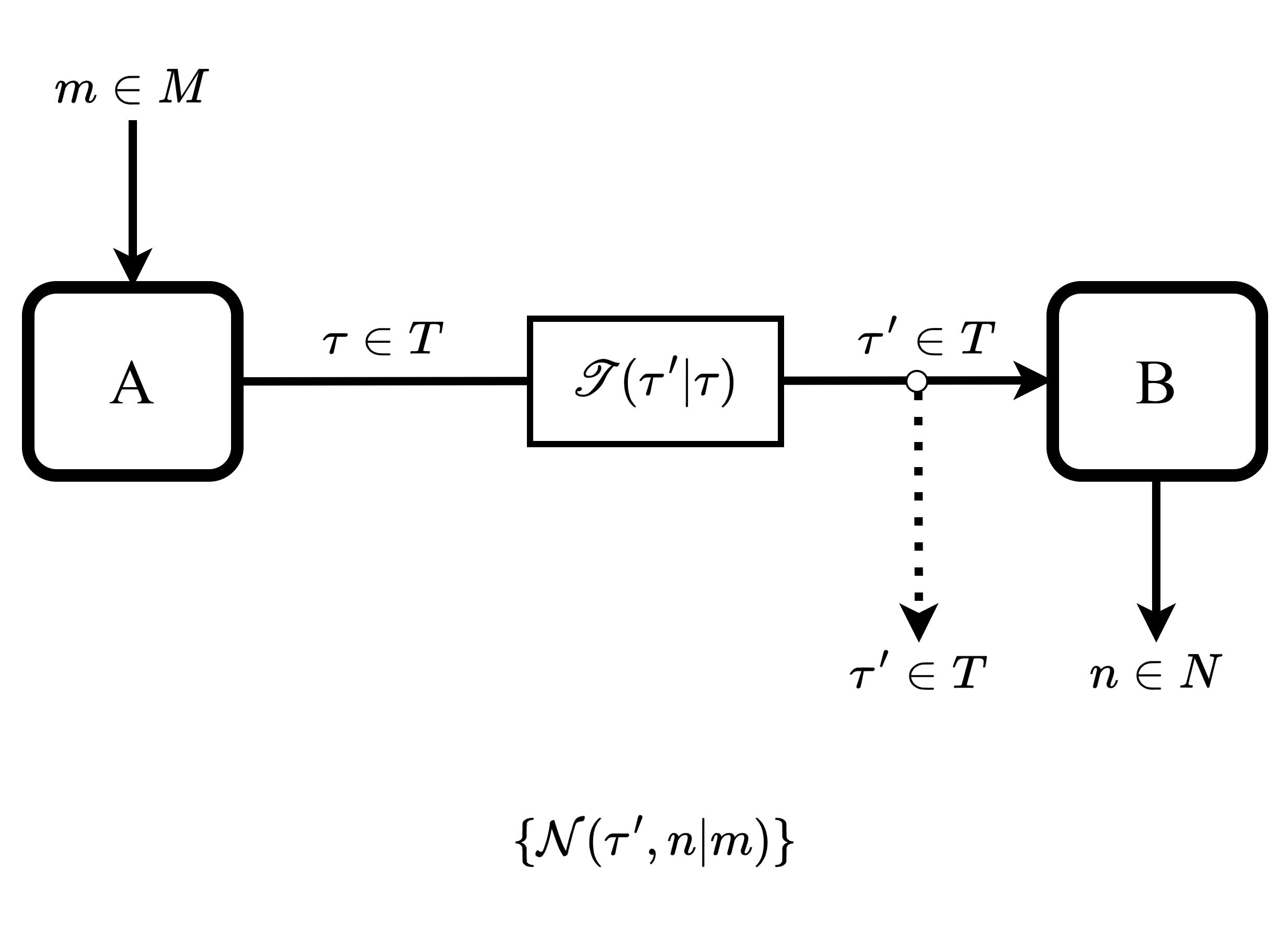}
    \caption{Schematic representation of the task $\mathbb{CW}_{M,N}[\mathscr{T},\{w^m_{\tau',n}\}]$ with wire-reading.}
    \label{fig: Wirereading}
\end{figure}

\begin{prop}\label{prop:wirereading}
For every task $\mathbb{CW}_{M,N}[\mathscr{T},\{w^m_{\tau',n}\}]$ defined with a noiseless channel $\mathscr{T}$, there exist a corresponding task $\mathbb{C}_{\tilde{M},\tilde{N}}[\mathscr{T},\{\tilde{w}^m_{n}\}]$, such that for every no-signalling correlation $P$ assistance, the optimal payoff $S^{opt}$ in the former task is equal the the optimal payoff $S_{W}^{opt}$ in the latter task while using the same correlation.
\end{prop}
\begin{proof}
    Please refer to appendix \ref{appendix:wire-reading} for the proof.
\end{proof}
As a consequence of Proposition \ref{prop:wirereading}, the shared randomness-assisted bound, entanglement-assisted bound and no-signalling correlation-assisted bounds for the payoff in task $\mathbb{CW}_{M,N}[\mathscr{T},\{w^m_{\tau',n}\}]$ are equal to the respective bounds for some task $\mathbb{C}_{\tilde{M},\tilde{N}}[\mathscr{T},\{\tilde{w}^m_{n}\}]$.\\

Wire-reading in communication tasks can reveal the advantage of entanglement assistance, which may not exist otherwise. This can be easily shown through a one-way communication task $\mathbb{C}_{M,N}[\mathscr{T},\{w^m_n\}]$ where Alice has access to a noiseless one-bit channel $\mathscr{T}:\{0,1\}\to \{0,1\}$, {\it i.e.}, $T=\{0,1\}$ and $\mathscr{T}(\tau'|\tau)=\delta_{\tau',\tau}$. The input set for Alice is $M=\{1,2,3,4\}$, and Bob gives output from the set $N=\{0,1\}$ (see Figure \ref{fig: Wirereading2}). For the success metric $S_{W}(\mathcal{N})=\sum_{m,n} w^m_{n} \mathcal{N}(n|m)$ in this case, the maximum payoff which is achievable when optimised over all shared correlations and strategies of Alice and Bob is given as $S_{W}^{alg}=\sum_{m\in M} [\max_{n\in N}\{w^m_{n}\}]$. We will refer to this quantity as the {\it algebraic maximum} for the payoff. For this task, $S_{W}^{alg}$  can be achieved while using shared randomness assistance to a noiseless bit channel. Upon receiving the input $m$, Alice can encode $\tau=\tilde{n}$ such $w^m_{\tilde{n}}=\max_{n\in N}\{w^m_{n}\}$ and Bob can output $n=\tau'=\tilde{n}$. Thus, the communication task described above cannot witness the advantage of assistance from any non-classical correlations to the communication channel.\\

\begin{figure}[h!]
\includegraphics[scale=0.095]{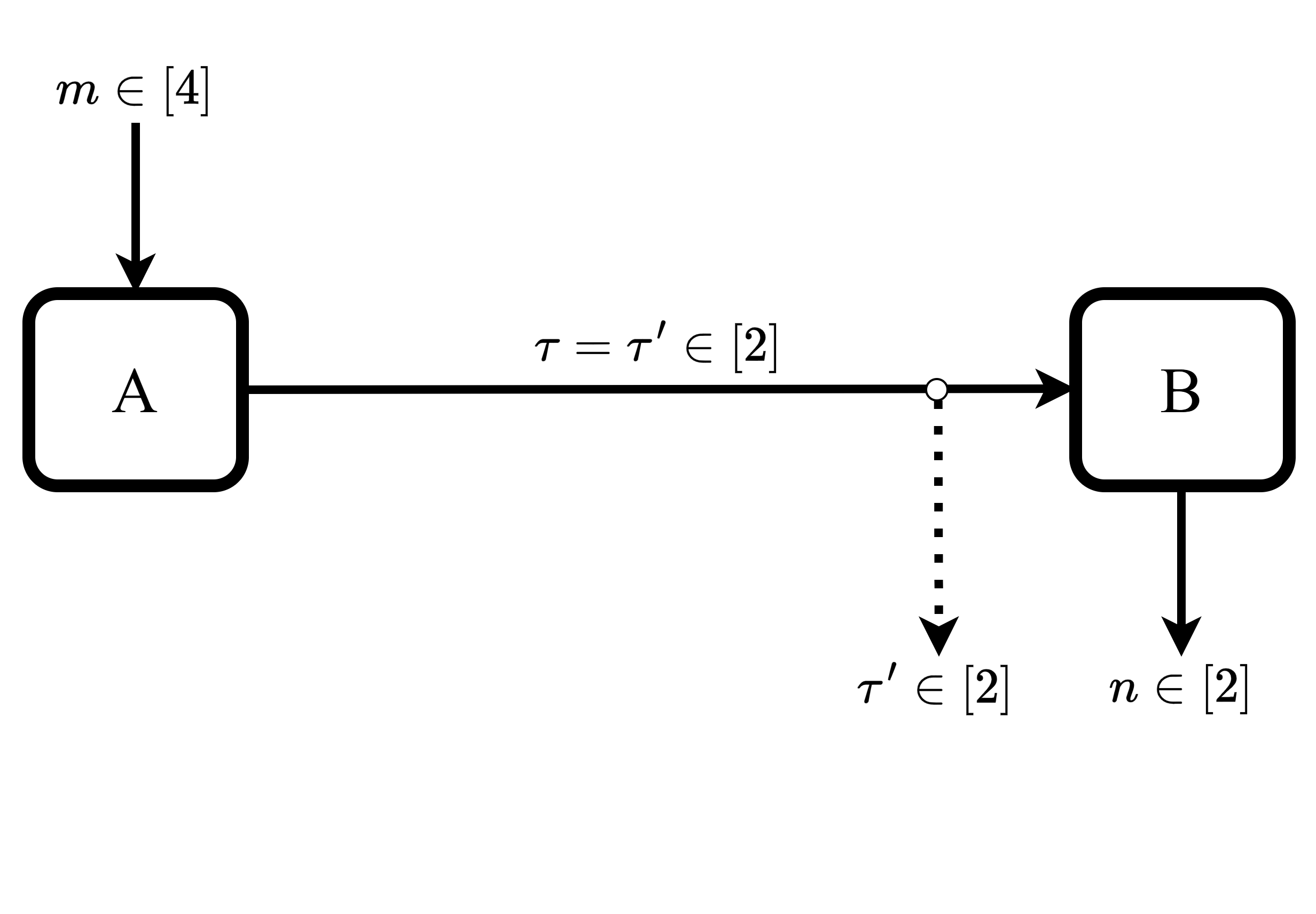}
    \caption{Schematic representation of a simple task with $1$ bit channel where wire-reading helps witness advantage of non-local correlations over shared randomness.}
    \label{fig: Wirereading2}
\end{figure}

In the same PM scenario, now consider a specific task $\mathbb{CW}_{M,N}[\mathscr{T},\{w^m_{\tau',n}\}]$ whose success metric is defined as in \eqref{eq: generic wirereding payoff}. The coefficients $w^m_{\tau',n}$ of the payoff function are defined in Table \ref{tab:wire reading example}. The algebraic maximum for the payoff is $S_{W}^{alg}= 1$. Using arbitrary shared randomness assistance to the classical channel, it can be shown that the maximum achievable payoff (local bound) is $\frac{3}{4}$. For some deterministic decoding strategy by Bob, say, $\mathbb{D}: T\to N$ where $n=D(\tau')=\tau'$,  an optimal deterministic encoding for Alice is given by $\mathbb{E}: M\to T$ where $E(1)=E(4)=0$ (or equivalently $1$), $E(2)=0$ and $E(3)=1$. This deterministic strategy yields a payoff $\frac{3}{4}$. Similarly, the maximum payoff while using any other deterministic encoding and decoding is upper bounded by $\frac{3}{4}$. As the payoff function is linear in conditional probabilities, the maximum payoff under shared randomness assistance can be achieved using an optimal deterministic encoding and decoding. Thus, the payoff for this task is upper bounded by $\frac{3}{4}$ when the parties have access to shared randomness. In the Section \ref{subsec: type0CS} (take $d=2,k=2$), we will show the advantage of assistance from non-local shared resources (Theorem \ref{theo: CS0 ns facet}, Corollary \ref{corr: CS0 noisy nl facet}) in this communication task.\\

\begin{table}[h!]
    \centering
    \begin{tabular}{ccccc}
        $m\backslash (\tau',n)$ & $(0,0)$ & $(0,1)$ & $(1,0)$ & $(1,1)$\\
        $1$ & $\frac{1}{4}$ & $0$ & $0$ & $\frac{1}{4}$\\
        $2$ & $\frac{1}{4}$ & $0$ & $\frac{1}{4}$ & $0$\\
        $3$ & $0$ & $\frac{1}{4}$ & $0$ & $\frac{1}{4}$\\
        $4$ & $0$ & $\frac{1}{4}$ & $\frac{1}{4}$ & $0$\\
    \end{tabular}
    \caption{The coefficients $\{w^{m}_{n,\tau'}\}$ in the payoff function for a task where Alice gets input $m\in M=\{1,2,3,4\}$ and Bob produces output $n\in N=\{0,1\}$. Here, $\tau'\in T=\{0,1\}$ is the received classical message using noiseless bit channel $\mathscr{T}:\{0,1\}\to \{0,1\}$.}
    \label{tab:wire reading example}
\end{table}

\section{Classical Communication Tasks Witnessing Non-locality}
Now, we introduce some families of one-way classical communication tasks based on wire-reading. We provide tight upper bounds on the payoff for these tasks when the classical communication is assisted by arbitrary shared randomness. We also show the advantage of shared entanglement and other non-local correlations in these tasks. The first family of tasks reveals the advantage of sharing any correlation on a non-local facet of the no-signalling polytope.

\subsection{Task: $\mathbb{CS}[d,k]$}\label{subsec: type0CS}
{\it Outline:} For the task $\mathbb{CS}[d,k]$, Alice is given a random function $
f:[d]\to[k]$. She is allowed to send an element $y\in[d]$ to Bob using a classical channel. Bob's goal is to output a value 
$b\in[k]$ such that $b\neq f(y)$. In other words, Bob must avoid outputting the image of the element sent by Alice under the function $f$. Since 
$f$ is a function, each element in $[d]$ has a unique image in $[k]$, which must be avoided for the task to win.\\

In the PM scenario, Alice receives randomly chosen input $m=(m_1,\cdots,m_d)\in M=\{1,\cdots,k\}^d$ where $d,k\in \mathbb{N}$ and $d,k\geq 2$. Bob's output $n\in N=[k]$. Alice has access to a noiseless classical communication channel $\mathscr{T}:[d]\to [d]$, {\it i.e.}, $T=[d]$ and $\mathscr{T}(\tau'|\tau)=\delta_{\tau',\tau}$. Let $\{\mathcal{N}(\tau',n|m)\}_{\tau',n,m)}$ represent the conditional probability distribution that Bob receives message $\tau'\in [d]$ and outputs $n\in N$ when Alice gets input $m\in M$. The success metric for this task is defined as
\begin{align}\label{eq: CS1 payoff}
&S_{W}(\mathcal{N})=\sum_{m,n,\tau'} w^m_{\tau',n} ~ \mathcal{N}(\tau',n|m)\nonumber\\
&\text{where } w^m_{\tau',n}=\frac{1}{k^d}(1-\delta_{n,m_{\tau'}})
\end{align}
Thus, given an input $m=(m_1,\cdots,m_d)$ for Alice, the parties get payoff $\frac{1}{k^d}$ if Bob's output $n\neq m_{\tau'}$ where $\tau'(=\tau)$ is the message sent/ received using the noiseless classical channel $\mathscr{T}$. We will refer to this task as $\mathbb{CS}[d,k]$. This task for $d=2,k=2$ is inspired by the task considered in \cite{Frenkel22} where Alice has $6$ possible inputs. The algebraic maximum for the payoff
\begin{align}\label{eq: CS1 payoffalg}
    S_{W}^{alg}&=\sum_{m} \max_{n\in N,\tau'\in [d]}\{w^m_{\tau',n}\}\nonumber\\
    &=\sum_{m\in M}\frac{1}{k^d}=1
\end{align}

For this task, we now provide a tight upper bound on the payoff when Alice and Bob have access to arbitrary shared randomness along with the communication channel.

\begin{thm}\label{theo: CS0 classical}
   For $\mathbb{CS}[d,k]$, the maximum payoff obtained while using shared randomness assistance to classical channel $\mathscr{T}:[d]\to [d]$ is $s_{\Lambda}=(1-\frac{1}{k^d})$.
\end{thm}

\begin{proof}
    Please refer to Appendix \ref{app:proof CS0 classical} for the proof.
\end{proof}

We will show the advantage of assistance from some non-local correlations for this task. First, we will prove the following property about no-correlations, which would be useful for later discussions.  

\begin{lem}\label{lemma:nlfacet correlation property}
    Given a function $\mathbb{LB}:[Y]\to [B]$ and a correlation $P_{NL}=\{P_{NL}(a,b|x,y)\}_{a\in [A],b\in [B],x\in [X],y\in [Y]}$ on the non-local facet of the no-signalling polytope $\mathcal{NS}$, there exists an input $x\in [X]$ such that for all $a\in [A]$ there exists a non-empty set $\Phi^{x,a}_{\mathbb{LB}}(P_{NL})=\{y\in [Y]:P_{NL}(a,b=\mathbb{LB}(y)|x,y)=0\}$.   
\end{lem}
\begin{proof}
   Please refer to Appendix \ref{app:lemma nlfacet correlation} for the proof. 
\end{proof}
\begin{thm}\label{theo: CS0 ns facet}
   For $\mathbb{CS}[d,k]$, assistance from any correlation $P_{NL}$ on a non-local facet of the no-signalling polytope $\mathcal{NS}:=\{P=\{P(a,b|x,y)\}_{a\in [A],b\in [k],x\in [X],y\in [d]}\}$ leads to the payoff $S_{W}^{alg}=1$.
\end{thm}
\begin{proof}
Please refer to Appendix \ref{app:CS0 nlfacet correlation} for the proof. 
\end{proof}
\begin{cor}\label{corr: CS0 noisy nl facet}
For any correlation $P_{NL}$ on a non-local facet of the no-signalling polytope $\mathcal{NS}$, if $(1-\frac{1}{k^{d-1}})<p\leq 1$ then assistance from correlation $P=pP_{NL}+(1-p)P_{WN}$ is advantageous over shared randomness assistance for the task $\mathbb{CS}[d,k]$.
\end{cor}
\begin{proof}
Please refer to Appendix \ref{app:CS0 noisy nlfacet correlation} for the proof.
\end{proof}

Notice that the shared randomness assisted bound on the payoff $s_{\Lambda}=(1-\frac{1}{k^d})$ for the task $\mathbb{CS}[d,k]$ converges to $S_{W}^{alg}=1$ polynomially in $k$ and exponentially in $d$. Consider a correlation $P_{NL}=\{P_{NL}(a,b|x,y)\}_{a\in [A],b\in [2],x\in [X],y\in [3]}$ which lies on a non-local facet of polytope $\mathcal{NS}$. For $\mathbb{CS}[d=3,k=2]$, while using a similar protocol as suggested in the proof of Corollary \ref{corr: CS0 noisy nl facet}, assistance from any correlation on the isotropic line joining correlation $P_{NL}$ and $P_{WN}$, {\it i.e.}, $P=pP_{NL}+(1-p)P_{WN}$ will not be advantageous over shared randomness if $p\leq \frac{3}{4}$. Next, we will introduce another family of communication tasks. For some of the tasks in this family, we will show that correlations which are described as $P=p P_{NL}+(1-p)P_{WN}$ (as before) will be advantageous if $\frac{1}{2}<p\leq 1$ (see Corllary \ref{corr:CS1xy22 noisy nl facet} and \ref{corr:13322 noisy nl facet}).\\ 

\subsection{Task: $\mathbb{CS}[\{P_{NL}^{*(i)}\}_{i\in I},d,k]$}\label{subsec: type1CS}
{\it Outline:} For the task $\mathbb{CS}[\{P_{NL}^{*(i)}\}_{i\in I},d,k]$, Alice is given a randomly chosen relation $R$ defined on the set $[d]\times [k]$ that is of certain type (satisfying condition (i) and (ii) discussed below). Alice is allowed to send an element $y\in [d]$  to Bob using the classical channel. Bob must output $b\in [k] $ such that it is an image of the element $y\in [d]$ under the relation $R$ to win. While we do not explicitly define all the relations that are possible inputs for Alice, it is implicitly defined next by starting with a function $\mathbb{LB}:[d]\to [k]$, invoking Lemma \ref{lemma:nlfacet correlation property} and compliance with conditions (i) and (ii) for this function. The condition (i) leads to a suboptimal payoff when assistance from shared randomness is allowed. The condition (ii) allows for a protocol which can be used to obtain maximum payoff when assistance from any correlation on the non-local facet defined by extremal points $\{P_{NL}^{*(i)}\}$ is accessible. Given a relation $R$, unlike a function, some element in $[d]$ may have multiple images in $[k]$, also some elements in $[d]$ {\it may not} have any image.\\

     For this task, Alice's input set $M$ is a set consisting of some relations $m\subset [d]\times [k]$, and Bob's output set $N=[k]$. To completely characterise the set of relations, let us introduce some of the notations that will be useful. As mentioned above in the outline, we would not explicitly define the set of relations but use some necessary conditions (i) and (ii) to define it.\\

Consider a non-local facet $Face\{P_{NL}^{*(i)}\}$ of the no-signalling polytope $\mathcal{NS}:=\{P=\{P(a,b|x,y)\}_{a\in [A],b\in [k],x\in [X],y\in [d]}\}$  $($with $B=k,Y=d)$ defined by the convex hull of some non-local extremal points $\{P_{NL}^{*(i)}\}_{i\in I}$. Here, $I$ denotes the subset of the indices labelling the extremal non-local correlations that span this facet.  Using Lemma \ref{lemma:nlfacet correlation property},
the following can be shown for non-local extremal points $\{P_{NL}^{*(i)}\}_{i\in I}$. Given a function $\mathbb{LB}:[d]\to [k]$, there exists an input $x^*\in [X]$, independent of the choice of $P_{NL}^{*(i)}$, such that $\forall~a\in [A]$ there exists a non-empty subset $\Phi^{x^*,a}_{\mathbb{LB}}=\bigcap_{j\in I}~\Phi^{x^*,a}_{\mathbb{LB}}(P_{NL}^{*(j)})~\subseteq [Y]$. For every function $\mathbb{LB}$, such an input $x^*\in[X]$ can be found as the contrary would imply that the Lemma \ref{lemma:nlfacet correlation property} does not hold for some correlations, say $P=\sum_{i\in I}\frac{1}{|I|}P_{NL}^{*(i)}$, on the non-local facet $Face\{P_{NL}^{*(i)}\}$ defined by these extremal correlations.\\ 

Thus, for all $y\in \Phi^{x^*,a}_{\mathbb{LB}}$ and every non-local correlation $P_{NL}^{*(i)}=\{P_{NL}^{*(i)}(a,b|x,y)\}_{a\in [A],b\in [k],x\in [X],y\in [d]}~$ where $i\in I$, 
\begin{align}\label{eq: defi phi xay}
     P_{NL}^{*(i)}&(a,b=\mathbb{LB}(y)|x^*,y)=0 
\end{align}

Additionaly, for each $y\in \Phi^{x^*,a}_{\mathbb{LB}}$, we would define the following set

\begin{equation}\label{eq: def phisxay}
   \Phi^{x^*,a,y}_{\mathbb{LB}}:=\{b\in [k]:P_{NL}^{*(i)}(a,b|x^*,y)>0 \text{ for some } i\in I\} 
\end{equation}

\begin{rem}\label{rem: empty phisxay}
    Notice that $\cup_{y\in\Phi^{x^*,a}_{\mathbb{LB}}}~\Phi^{x^*,a,y}_{\mathbb{LB}}= \varnothing$ only if $P_{NL}^{*(i)}(a|x^*)=0$ for all $i\in I$. Given $\Phi^{x^*,a,y}_{\mathbb{LB}}=\varnothing$ for $y\in\Phi^{x^*,a}_{\mathbb{LB}} $, $\sum_b P_{NL}^{*(i)}(a,b|x^*,y)=0=P_{NL}^{*(i)}(a|x^*)$ for all $i\in I$.\\
\end{rem}

Using these sets $\Phi^{x^*,a}_{\mathbb{LB}}$ and $\Phi^{x^*,a,y}_{\mathbb{LB}}$ defined above, we consider all relations $R_{\mathbb{LB}}\subset [d]\times [k]$ that satisfy the following two properties: 

\begin{itemize}\label{condition task 2}
    \item[(i)] $(y,b)\in R_{\mathbb{LB}}\implies$ $b\in [k]\backslash\{\mathbb{LB}(y)\}$.\\
    \item[(ii)] for all $a\in [A]$ such that $\cup_{y\in\Phi^{x^*,a}_{\mathbb{LB}}}~\Phi^{x^*,a,y}_{\mathbb{LB}}\neq \varnothing$, there exists $y\in \Phi^{x^*,a}_{\mathbb{LB}}$, for which $(y,b)\in R_{\mathbb{LB}}$ for all $b\in \Phi^{x^*,a,y}_{\mathbb{LB}}$.\\
\end{itemize}

$\mathcal{R}_{\mathbb{LB}}$ denotes the set of all such relations for each function $\mathbb{LB}$ and $\mathcal{R}=\cup_{\substack{\mathbb{LB}:[d]\to [k]}}\mathcal{R}_{\mathbb{LB}}$.\\ 

We will now define a task specific to the non-local facet $Face\{P_{NL}^{*(i)}\}$. For the task $\mathbb{CS}[\{P_{NL}^{*(i)}\}_{i\in I},d,k]$, Alice randomly receives some input $m\in M=\mathcal{R}$, {\it i.e.}, a relation as defined above. Bob gives some output $n\in N=[k]$. Alice can access a noiseless classical channel $\mathscr{T}:[d]\to [d]$ in this task, {\it i.e.}, $T=[d]$ and $\mathscr{T}(\tau'|\tau)=\delta_{\tau',\tau}$. Say $\{\mathcal{N}(\tau',n|m)\}_{\tau',n,m}$ represents the conditional probability that Bob receives message $\tau'\in [d]$ and outputs $n$ when Alice receives input $m$. The payoff for this task is given by

\begin{align}\label{eq: CS2 payoff}
&S_{W}(\mathcal{N})=\sum_{m,n,\tau'} w^m_{\tau',n} ~ \mathcal{N}(\tau',n|m)\nonumber\\
&\text{ where }  w^m_{\tau',n}= \begin{cases}
    &\omega_m \text{ if } (\tau',n)\in m\\
    &0 ~~~\text{ otherwise }
\end{cases}
\end{align}
Here, $\omega_m\geq 0$, $\sum_{m\in M} \omega_m=1$. Additionally, in this task for each function $\mathbb{LB}:[d]\to [k]$ there exists input $m\in \mathcal{R}_{\mathbb{LB}}$ such that $\omega_m> 0$. The algebraic maximum for the payoff, which can be obtained using some protocol and communication resource, for this task is $S_{W}^{alg}=\sum_{m\in M} \omega_m=1$. We give a tight upper bound on the payoff for the task when Alice and Bob have access to shared randomness, which can assist the classical communication channel. Note that, for the task, if $\omega_m= 0$ for some input $m\in M$, then the input does not contribute to the payoff and therefore can be dropped from the input set without significantly altering the task.

\begin{thm}\label{theo: CS1 classical bound}
   In $\mathbb{CS}[\{P_{NL}^{*(i)}\}_{i\in I},d,k]$, the maximum payoff obtained while using arbitrary shared randomness assistance is
   \begin{equation}
       s_{\Lambda}=\max_{\mathbb{D}:[d]\to [k]}\left[\sum_{\substack{m\in M:\\
     \exists \alpha\in[d] \land (\alpha,\mathbb{D}(\alpha))\in m}} \omega_m\right]<1\nonumber
   \end{equation}
\end{thm}
\begin{proof}
    Please refer to Appendix \ref{app:CS1 classical bound} for the proof.
\end{proof}
We will show the advantage of assistance from some non-local correlations for this task.\\

\begin{thm}\label{theo: CS1 nl facet}
   For $\mathbb{CS}[\{P_{NL}^{*(i)}\}_{i\in I},d,k]$, assistance from any correlation $P_{NL}\in Face\{P_{NL}^{*(i)}\}$ on the non-local facet of the no-signalling polytope $\mathcal{NS}:=\{P=\{P(a,b|x,y)\}_{a\in [A],b\in [k],x\in [X],y\in [d]}\}$ leads to the payoff $S_{W}^{alg}=1$.
\end{thm}
\begin{proof}
    Please refer to Appendix \ref{app:CS1 nl facet} for the proof.
\end{proof}

We will now explicitly consider some communication tasks corresponding to a no-signalling extremal correlation with dichotomic outputs. 
\subsubsection{Communication task for extremal correlation with dichotomic outputs}
In the simplest non-local scenario with dichotomic inputs and outputs, the polytope $\mathcal{NS}:=\{P=\{P(a,b|x,y)\}_{a\in [2],b\in [2],x\in [2],y\in [2]}\}$ has $8$ non-local extremal correlations or PR-boxes $P_{NL}^{*(\alpha,\beta,\gamma)}=\{P_{NL}^{*(\alpha,\beta,\gamma)}(a,b|x,y)\}_{a\in [2],b\in [2],x\in [2],y\in [2]}$ where:
\begin{align}\label{eq:PR boxes}
    P_{NL}^{*(\alpha,\beta,\gamma)}(a,b|x,y)=\frac{\delta_{(a-1) \oplus (b-1),f(x,y,\alpha,\beta,\gamma)}}{2}.
\end{align}
$f(x,y,\alpha,\beta,\gamma)=(x-1) (y-1) \oplus \alpha (x-1) \oplus \beta (y-1) \oplus \gamma$ and $(\alpha,\beta,\gamma)\in \{0,1\}^3$.\\

As a generic case, consider the no-signalling polytope   $\mathcal{NS}:=\{P=\{P(a,b|x,y)\}_{a\in [2],b\in [2],x\in [X],y\in [d]}\}$ where corresponding to each input there are two outputs in the non-local scenario $($with $Y=d)$. Let $P_{NL}^{*}=\{P_{NL}^{*}(a,b|x,y)\}_{a\in [2],b\in [2],x\in [X],y\in [d]}$ denote some non-local extremal correlation of $\mathcal{NS}$ such that for every $y',y''\in [d]$, there is  $x',x''[X]$ for which $\{P_{NL}^{*}(a,b|x,y)\}_{a\in[2],b\in [2],x\in\{x',x''\},y\in\{y',y''\}}$ is equivalent to a two-input two-output PR-boxes as in equation \eqref{eq:PR boxes}. 

Corresponding to the extremal correlation $P_{NL}^{*}$, we describe the task $\mathbb{CS}[\{P_{NL}^{*}\},d,k=2]$. Alice's input set $M$ consists of relation $\mathrm{R}_{ijkl}:=\{(i,k),(j,l)\}\subset [d]\times [2]$ where $i<j$, $i,j\in [d]$ and $k,l\in [2]$. Bob's output set $N=[2]$. Alice has access to a noiseless channel $\mathscr{T}:[d]\to [d]$. The success metric is as given in \eqref{eq: CS2 payoff} and coefficient $\omega_m=\frac{1}{2d(d-1)}$. Note we dropped all other relations not listed above from the input set; alternately, we assign a payoff of $0$ for them. For this task, we will now discuss the shared randomness bound on the payoff and show that assistance from the extremal correlation $P_{NL}^{*}$ leads to the maximum achievable payoff. 

\begin{cor}\label{corr:CS1 xy22 srbound}
    In $\mathbb{CS}[\{P_{NL}^{*}\},d,k=2]$, $(\mathrm{i})$ the maximum payoff obtained while using arbitrary shared randomness assistance is $s_\Lambda=\frac{3}{4}$ and $(\mathrm{ii})$ the correlation $P_{NL}^{*}$ leads to the payoff $S_{W}^{alg}=1$.
\end{cor}
\begin{proof}
    Please refer to Appendix \ref{app:CS1 classical bound xy22} for the proof of part (i) and Appendix \ref{app:CS1 nl facet bound xy22} for the proof of part (ii).
\end{proof}

\begin{cor}\label{corr:CS1xy22 noisy nl facet}
   If $\frac{1}{2}<p\leq 1$ then assistance from correlation $P=pP_{NL}^{*}+(1-p)P_{WN}$ is advantageous over shared randomness assistance for the task $\mathbb{CS}[\{P_{NL}^{*}\},d,k=2]$.
\end{cor}
\begin{proof}
    Please refer to Appendix \ref{app:CS1 noisy nl facet bound xy22} for the proof.
\end{proof}

\subsubsection{Communication task for extremal correlations violating $I3322$ inequality}\label{subsubsec: i3322}
We will now consider a specific instance of the task $\mathbb{CS}[\{P_{NL}^{*(i)}\}_i,d,k]$ discussed above, for a non-local facet of no-signalling polytope $\mathcal{NS}:=\{P=\{P(a,b|x,y)\}_{a\in [2],b\in [2],x\in [3],y\in [3]}\}$.  Consider non-local facet $Face\{P_{NL}^{*(1)}, P_{NL}^{*(2)}\}$. Extremal correlation $P_{NL}^{*(1)}=\{P_{NL}^{*(1)}(a,b|x,y)\}_{a\in [2],b\in [2],x\in [3],y\in [3]}$ , where $P_{NL}^{*(1)}(a,b|x,y)=\frac{\delta_{a,b}}{2}$ if $(x,y)\in\{(1,1),(1,2),(2,1),(2,2),(2,3),(3,2),(3,3)\}$ else $P_{NL}^{*(1)}(a,b|x,y)=\frac{1-\delta_{a,b}}{2}$. For extremal correlations $P_{NL}^{*(2)}=\{P_{NL}^{*(2)}(a,b|x,y)\}_{a\in [2],b\in [2],x\in [3],y\in [3]}$, $P_{NL}^{*(2)}(a,b|x,y)=\frac{\delta_{a,b}}{2}$ if $(x,y)\in\{(1,1),(1,2),(2,1),(2,2),(2,3),(3,2)\}$ else $P_{NL}^{*(2)}(a,b|x,y)=\frac{1-\delta_{a,b}}{2}$.
Correlation on this facet $Face\{P_{NL}^{*(1)}, P_{NL}^{*(2)}\}$ maximally violates the $I3322$ Bell inequality given as follows:

\begin{align}\label{eqn:i3322Bellineq}
\mathbf{B}(P)=&-P_A(2|2)-P_B(2|1)-2P_B(2|2)+\nonumber\\
&P(2,2|1,1)+P(2,2|1,2)+P(2,2|2,1)+\nonumber\\
&P(2,2|2,2)-P(2,2|1,3)+P(2,2|2,3)-\nonumber\\
&P(2,2|3,1)+P(2,2|3,2)
\end{align}

where marginal correlation $P_A(a|x)=\sum_{b\in [B]}P(a,b|x,y)$ for $a\in [A]$ and $x\in [X]$. Similarly, marginal correlation $P_B(b|y)=\sum_{a\in [A]}P(a,b|x,y)$ for $b\in [B]$ and $y\in [Y]$. The local bound for the inequality is $\beta_L=0$, and the no-signalling maximum violation is $1$.\\

Corresponding to the facet $Face\{P_{NL}^{*(1)}, P_{NL}^{*(2)}\}$, we describe the task $\mathbb{CS}[\{P_{NL}^{*(1)},P_{NL}^{*(2)}\},d=3,k=2]$. Alice's input set consists $M$ consists of relation $\mathrm{R}_{ijkl}:=\{(i,k),(j,l)\}\subset [3]\times [2]$ where $i<j$, $i,j\in [3]$ and $k,l\in [2]$. Bob's output set $N=[2]$. Alice has access to noiseless channel $\mathscr{T}:[3]\to [3]$. The success metric is as given in \eqref{eq: CS2 payoff} and coefficient $\omega_m=\frac{1}{12}$. Note we dropped all other relations not listed above from the input set; alternately, we assign a payoff of $0$ for them. For this task, using Corollary \ref{corr:CS1 xy22 srbound}, the shared randomness bound on the payoff is $s_\Lambda=\frac{3}{4}$.  We will now show the advantage of the assistance from some no-signalling correlation as well as shared entanglement. 

\begin{cor}\label{corr:13322 nl facet}
    In $\mathbb{CS}[\{P_{NL}^{*(1)},P_{NL}^{*(2)}\},d=3,k=2]$, assistance from any correlation $P\in Face\{P_{NL}^{*(1)}, P_{NL}^{*(2)}\}$ leads to the payoff $S_{W}^{alg}=1$.
\end{cor}
\begin{proof}
    Please refer to Appendix \ref{app:CS1 nl facet bound I322} for the proof.
\end{proof}
\begin{cor}\label{corr:13322 noisy nl facet}
    For $P\in Face\{P_{NL}^{*(1)}, P_{NL}^{*(2)}\}$, if $\frac{1}{2}<p\leq 1$ then assistance from correlation $P=pP_{NL}+(1-p)P_{WN}$ is advantageous over shared randomness assistance for the task $\mathbb{CS}[\{P_{NL}^{*(1)},P_{NL}^{*(2)}\},d=3,k=2]$.
\end{cor}
\begin{proof}
    The proof is similar to the proof of Corollary \ref{corr:CS1xy22 noisy nl facet}
\end{proof}

\subsection{Bounding quantum advantage}
In this section, we demonstrate the advantage of quantum non-locality in the tasks introduced in Sections \ref{subsec: type0CS} and \ref{subsec: type1CS}. To estimate the maximum payoff achievable with entanglement-assisted protocols, we use an alternate semi-definite programming (see-saw) technique to recover lower bounds on the maximal quantum violation of the associated Bell inequality and the advantageous quantum strategy achieving it. For specific values of $d$ and $k$ in the task $\mathbb{CS}[d,k]$, we compute the quantum correlation-assisted payoff using the see-saw algorithm with qubits, observing values that surpass the classical (local) bound (see Table~\ref{tab: CSdk see-saw quantum}). For the task $\mathbb{CS}[\{P_{NL}^{*(i)}\}_{i\in I},d,k]$ presented in Section \ref{subsubsec: i3322}, we found that quantum correlation-assisted payoff using the see-saw algorithm with qubits achieves $0.85355$ which is higher than the local bound $0.75$ on the payoff.\\

We employ the NPA hierarchy~\cite{Navascues2007,Navas_2008} to retrieve upper bounds on the maximum quantum violation for several cases in Table \ref{tab: CSdk npa quantum}. For the task $\mathbb{CS}[d,k]$, in the simplest case of when $d=2,k=2$, the upper bound $0.85355$ from the level $2$ NPA hierarchy match the lower bound from qubits up to machine precision, and hence we recover the maximal quantum violation. Intriguingly, for the two cases $d=2,k=3$ and $d=3,k=2$, the bound for the NPA hierarchy level $1+AB$ deviates from the qubit lower bounds, which we investigate further.\\ 

In particular, we identify a task $\mathbb{CS}[d=2, k=3]$ where a two-qutrit strategy with Alice's measurements being non-projective offers a clear advantage over a two-qubit quantum strategy. To bound the performance achievable with a two-qubit strategy, we employ the Navascués–Vértési method \cite{NVMethod}, which provides an upper bound of $0.93491$, saturating the see-saw lower bound, on the quantum payoff with two-qubit entanglement and rank $1$ projectors. Finally, using the see-saw method, we recover a two-qutrit strategy which yields a higher payoff than qubits and saturates the upper-bound $0.93883$ from level $2$ of the NPA hierarchy, hence forms the maximal quantum payoff. For the $d=3,k=2$ case, although NPA level $2$ recovers a tighter bound $0.95209$ but it still does not match the see-saw lower bounds. Additionally, the lower bounds obtained from the see-saw algorithm did not increase up to local dimension  $8$. For the task $\mathbb{CS}[\{P_{NL}^{*(i)}\}_{i\in I},d,k]$ presented in Section \ref{subsubsec: i3322}, the upper bound $0.85355$ from the level $2$ NPA hierarchy match the lower bound from two-qubit entangled state up to machine precision, and thus we recover the maximal quantum violation. Interestingly, the task $\mathbb{CS}[d=2, k=3]$ enables device-independent certification of the local dimension of the shared entangled state. 

\begin{table}[h!]
    \centering
    \begin{tabular}{|c|c|c|c|}
    \hline
        $(d,k)$ & Quantum & Local  & Advantage\\
         & payoff &  Bound & \\
     \hline    
        $(2,2)$ & $0.85355^{(i)}$ & $0.75$ & $0.10355$\\
     \hline     
        $(3,2)$ & $0.94975$ & $0.875$ & $0.07475$\\
    \hline    
        $(4,2)$ & $0.97487$ & $0.9375$ & $0.03737$\\
     \hline     
        $(5,2)$ & $0.99235$ & $0.96875$ & $0.02360$\\
     \hline     
        $(6,2)$ & $0.99371$ & $0.98437$ & $0.00934$\\
     \hline     
        $(2,3)$ & $0.93491^{(ii)}$ & $0.88888$ & $0.04603$\\
     \hline     
        $(2,4)$ & $0.96338$ & $0.9375$ & $0.02588$\\
     \hline     
        $(2,5)$ & $0.97656$ & $0.96$ & $0.01656$\\
     \hline     
        $(2,6)$ & $0.98373$ & $0.97222$ & $0.01151$\\
     \hline     
        $(3,3)$ & $0.98511$ & $0.96296$ & $0.02215$\\ 
    \hline    
    \end{tabular}
    \caption{Advantage of quantum correlation assistance using {\it two-qubit} entanglement in the task $\mathbb{CS}[d,k]$ obtained using the see-saw algorithm for different values of $d,k$. The payoff from two-qubit entanglement assistance in $(i)$ matches the upper bound on the entanglement-assisted payoff from the level $2$ NPA hierarchy (see Table \ref{tab: CSdk npa quantum}), and that in $(ii)$ matches the upper bound on payoff obtained from the Navascués–Vértési method \cite{NVMethod} for two-qubit entanglement and rank $1$ projectors.}
    \label{tab: CSdk see-saw quantum}
\end{table}
\begin{table}[h!]
    \centering
    \begin{tabular}{|c|c|}
    \hline
        $(d,k)$ & Maximum Quantum Payoff\\
         & (upper bound) \\
        \hline
        $(2,2)$ & $0.85355^{(i)}$ (NPA level $2$) \\
        \hline
        $(3,2)$ & $0.96231$ (NPA $1+AB$)\\
        \hline
        $(3,2)$ & $0.95209$ (NPA level $2$)\\
        \hline
        $(2,3)$ & $0.93883^{(ii)}$ (NPA level $2$) \\
        \hline
    \end{tabular}
    \caption{Maximum advantage of quantum correlation assistance in the task $\mathbb{CS}[d,k]$ obtained using NPA hierarchy for different values of $d,k$. The upper bound in $(i)$ matches the lower bound on payoff obtained using two-qubit entanglement (see Table \ref{tab: CSdk see-saw quantum}). The upper bound in
    $(ii)$ matches the lower bound on the payoff obtained using the see-saw algorithm for two-qutrit entanglement, while two-qubit entanglement strategies yield lower payoff (see Table \ref{tab: CSdk see-saw quantum}).}
    \label{tab: CSdk npa quantum}
\end{table}

\section{Summary and Discussion}
In this work, we explored some fundamental connections between non-locality and its advantages in correlation-assisted classical communication tasks in the PM scenario. We construct a Bell inequality tailored to any bounded classical communication task where assistance from no-signalling correlation is allowed using the wire-cutting technique. We prove that a violation of the associated Bell inequality by a non-local correlation implies an advantage in the communication task, thereby connecting non-locality to efficiency in the communication task.\\

We further discussed wire-reading, a complementary tool that uses the readability of classical messages to demonstrate quantum advantage even in scenarios where classical strategies appear optimal under the standard setting. Building on these tools, we constructed families of classical communication tasks where non-local correlations outperform shared randomness. An advantage in these tasks is closely tied to the geometric structure of the no-signalling polytope. We show that any correlation from a non-local facet of the no-signalling polytope ensures optimal payoff in some communication tasks, and we explicitly analysed a task assisted by extremal non-local correlations with dichotomic outputs. We report several instances of quantum advantage and identify an example where qutrit entanglement outperforms qubit entanglement in achieving maximum payoff.\\

Our work contributes to growing research demonstrating the advantage of shared entanglement and other non-local resources in assisting classical communication in various setups. Early studies revealed that shared entanglement can outperform shared randomness in multi-party communication complexity scenarios \cite{Buhrman001,Cleve1997,PhysRevA.60.2737}. Subsequent works established links between Bell inequality violations and advantages in such tasks \cite{Brukner2002,BRUKNER2003,Brukner2004,Tavakoli2017,Tavakoli2020doesviolationofbell,Ho2022}, highlighting the operational significance of non-local correlations. Quantum advantages have also been demonstrated in tasks inspired by random access codes \cite{RAC1,RAC2,Tavakoli2021, Pauwels2022,Pauwels22}, where entanglement-assisted strategies achieve higher payoff than shared randomness-assisted ones. Beyond these scenarios, entanglement and no-signalling correlations have been shown to increase the zero-error capacity of classical channels and to reduce the size of noiseless classical channels required for simulating noisy ones \cite{Cubitt2010,Cubitt2011,Leung2012}.  In \cite{Yadavalli2022contextualityin}, a related task involving a noisy classical channel was introduced, where the advantage of entanglement over shared randomness was linked to quantum contextuality.  In a PM scenario, \cite{Frenkel22} presented a task where one-bit communication assisted by shared randomness is suboptimal, and any correlation violating the CHSH inequality provides a quantum advantage. Further works \cite{agarwal2025nonlocalityassistedenhancementerrorfreecommunication} showed that quantum and no-signalling correlations can enhance the zero-error capacity of classical channels with zero unassisted zero-error capacity. In \cite{Alimuddin2023}, the authors demonstrated an advantage of Hardy-type non-local correlations, with only dichotomic inputs and outputs, in assisting one-bit communication. Remarkably, an instance of unbounded separation between entanglement-assisted and shared randomness–assisted classical communication was established in \cite{Perry2015}.\\

In contrast to most prior work, except \cite{Cubitt2010,Cubitt2011,Frenkel22,Alimuddin2023,agarwal2025nonlocalityassistedenhancementerrorfreecommunication}, we consider a minimal prepare and measure scenario where Bob does not have input. Using the wire-cutting technique, we construct a Bell inequality for each communication task. More significantly, employing the wire-reading technique, we explicitly construct a communication task tailored to each non-local facet of the no-signalling polytope which were not previously known in the literature.\\

Non-local correlations are known to enhance classical communication, yet it is still unclear whether all non-locality guarantees such an advantage. Our work builds up to the broader question of whether non-locality is both necessary and sufficient for a quantum communication advantage. While wire-cutting (Theorem \ref{theo: PMtoBell}) definitively answers the former, the latter question, whether every non-local correlation can yield an advantage, remains an open question. The wire-reading technique enables the construction of classical communication tasks that are closely tailored to a given non-local resource, providing a versatile and powerful tool for probing the sufficiency of non-locality in such tasks.  Section \ref{subsec: type1CS} presents only a few examples of tasks for a non-local facet of the no-signalling polytope, leaving many similar possibilities unexplored. Modifying and extending our constructions for different facets and analysing their performance when the communication channel could be potentially noisy are promising directions for future investigation.

\section{Acknowledgements}
We thank Mir Alimuddin and Marcin Pawłowski for insightful discussions. PH and SR acknowledge support by the National Science Centre
in Poland under the research grant Maestro (2021/42/A/ST2/00356). AC acknowledges financial support by NCN Grant SONATINA 6 (Contract No. UMO-2022/44/C/ST2/00081). SSB acknowledges funding by the Spanish MICIN (project PID2022-141283NB-I00) with the support of FEDER funds, and by the Ministry for Digital Transformation and of Civil Service of the Spanish Government through the QUANTUM ENIA project call- Quantum Spain project, and by the European Union through
the Recovery, Transformation and Resilience Plan - NextGeneration EU within the framework of the Digital Spain 2026 Agenda.

\bibliographystyle{quantum}
\bibliography{bib}

\onecolumn
\appendix

\section{Bob to Alice signalling correlation lead to invalid conditional probabilities in \eqref{eq: NSRealizableChannel}}\label{appendix: Invalid for signalling}
Say the correlation $P=\{ P(a,b|x,y)\}_{a\in[A],b\in[B],x\in[X],y\in [Y]}$ is signalling from Bob to Alice. We will consider, $P(a,b|x,y)=0$ if $a\notin[A]$ or $b\notin [B]$. As $P$ is Bob to Alice signalling, there exists some $x^*\in [X]$ such that for some $a\in [A]$, there are $y',y''\in [Y]$ for which $P_A(a|x^*,y')\neq P_A(a|x^*,y'')$. Clearly, $\sum_{a\in [A]}P_A(a|x^*,y)=1$ for all $y\in [Y]$. Also we can define a set $\{y_1,y_2,\cdots, y_A\}$ such that for each $a\in [A]$, $y_a\in [Y]$ and $P_A(a|x^*,y_a)=\max_{y\in [Y]}P_A(a|x^*,y) $. Notice that for the signalling correlation $P$, $y_1=y_2=\cdots=y_A=y'''\in[Y]$ is not possible as it will imply $\sum_{a\in [A]}P_A(a|x^*,y)<1$ for some $y\in[Y]$ where $y\neq y'''$. Also, $\sum_{a\in[A]} P_A(a|x^*,y_a)>\sum_{a\in[A]} P_A(a|x^*,y_1)=1$.\\

Now, consider the set $M=[X]$ and $N=[B]$ for instance. And the channel $\mathscr{T}:T\to T$ where $T=[\max\{A,Y\}]$. For $\tau\in [A]$, the conditional probability for channel is $\mathscr{T}(\tau'|\tau)=\delta_{\tau',y_{\tau}}$ and for $\tau\notin [A]$, $\mathscr{T}(\tau'|\tau)=\delta_{\tau',1}$. Also, $\{p(x|m)=\delta_{x,m}\}_{x,m}$ and $\{p(y|\tau')\}_{\tau', y}$ denotes coding respectively for Alice and Bob for choosing their input to $P$ based on their input/ received message. Here $p(y|\tau')=\delta_{y,\tau'}$ for $\tau'\in [Y]$ else $p(y|\tau')=\delta_{y,1}$. $\{p_{e}(\tau|a,m)=\delta_{\tau,a}\}_{\tau, a, m}$ is Alice's encoding for the communication based on her input $m$ and response from correlation $P$. For $\{p_{d}(n|\tau',b)=\delta_{n,b}\}_{\tau', n, b}$ is Bob's decoding based on his received message and response from resource correlation. Now substituting these in equation $3$ for a particular $m=x^*$, we get the following.

\begin{equation} \label{eq: signalling}
 \begin{split}
     \mathcal{N}(n|m=x^*) = \sum_{\substack{a\in [A], b\in[B], x\in[X], y\in [Y],\\ \tau, \tau'\in [A]}}
     \big[ \delta_{x,m} \,\delta_{\tau,a} \, P(a,b|x,y)\delta_{\tau',y_{\tau}} \, \delta_{y,\tau'} \, \delta_{n,b} \big]= \sum_{\substack{a\in [A]}}
     \big[  P(a,n|x^*,y_a) \big]\
 \end{split}
\end{equation}
As $P(a,b|x,y)=0$ for $a\notin [A]$ and therefore the terms vanish for  $\tau\in T\backslash [A]$. Now,
\begin{equation} \label{eq: signalling2}
 \begin{split}
     \sum_{n\in[B]}\mathcal{N}(n|m=x^*) =\sum_{n\in[B]} \sum_{\substack{a\in [A]}}
     \big[  P(a,n|x^*,y_a) \big]= \sum_{\substack{a\in [A]}}
     \big[  P_{A}(a|x^*,y_a) \big]>1
 \end{split}
\end{equation}
Thus, $\mathcal{N}(n|m=x^*)$ is not a valid conditional probability distribution.

\section{Proof of Proposition \ref{prop:wirereading}}\label{appendix:wire-reading}
\begin{proof}
For a task $\mathbb{CW}_{M,N}[\mathscr{T},\{w^m_{\tau',n}\}]$, $M $, $N$ respectively denote the input and output set of Alice and Bob, $\mathscr{T}: T \to T$ is the resource classical channel which is noiseless, {i.e.}, $\mathscr{T}(\tau'|\tau)=\delta_{\tau',\tau}$. The payoff while using wire reading is given as
\begin{equation}\label{eq:app_proofwirereading1}
    S_W(\mathcal{N'})=\sum_{m\in M,n\in N,\tau'\in T} w^m_{\tau',n} ~ \mathcal{N}(\tau',n|m)
\end{equation}
$\{\mathcal{N}(\tau',n|m)\}_{\tau',n,m}$ represents the conditional probability distribution of Bob receiving message $\tau'(=\tau)$ and giving output $n$ when Alice gets input $m$. Using the assistance from some shared no-signalling correlation $P=\{p(a,b|x,y)\}_{a\in [A],b\in [B],x\in [X],y\in [Y]}$, the optimal payoff for the task be denoted as $S_{W}^{opt}$. Since the payoff function is linear in conditional probabilities $\mathcal{N}(\tau',n|m)$, this payoff $S_{W}^{opt}$ can be achieved while using some deterministic strategy as follows. Alice uses the function $\mathbb{FX}:M\to [X]$ to decide the query $x=\mathbb{FX}(m)$ to the correlation $P$ conditioned on the input $m\in M$. Using the function $\mathbb{FT}:M\times [A]\to T$, Alice encodes message $\tau=\mathbb{FT}(m,a)$ based on $m$ and output $a\in [A]$ from the correlation $P$. Similarly, Bob uses the function $\mathbb{FY}:T\to [Y]$ to query $y=\mathbb{FY}(\tau')$ to the correlation $P$ upon receiving message $\tau' \in T$. Using the output $b\in [B]$ based on the query, Bob outputs $n=\mathbb{FN}(\tau',b)$ according to function $\mathbb{FN}:T\times [B]\to N$. The obtained conditional probability distribution $\{\mathcal{N}_{P}~ (\tau',n|m)\}_{\tau',n,m}$, thus obtained, can be written as 

\begin{align} \label{eq:app_proofwirereading2}
\mathcal{N}_P~(\tau' n|m) &= \sum_{\substack{a, b, x, y,\\ \tau}}  \big[ \delta_{x,\mathbb{FX}(m)}~\delta_{\tau,\mathbb{FT}(m,a)}~  P(a,b|x,y)~ \mathscr{T}(\tau'|\tau) \delta_{y,\mathbb{FY}(\tau')}~\delta_{n,\mathbb{FN}(\tau',b)}\big]\nonumber\\
&=\sum_{\substack{a,b, x, y}}  \big[ \delta_{x,\mathbb{FX}(m)}~\delta_{\tau',\mathbb{FT}(m,a)}~  P(a,b|x,y)~  \delta_{y,\mathbb{FY}(\tau')}~\delta_{n,\mathbb{FN}(\tau',b)}\big]
\end{align}

$S_{W}^{opt}=\sum_{m\in M, \tau'\in T,n\in N}w^m_{\tau',n}\mathcal{N}_P~(\tau' n|m)$. Corresponding to the task $\mathbb{CW}_{M, N}[\mathscr{T},\{w^m_{\tau',n}\}]$, consider the following task $\mathbb{CT}_{\tilde{M},\tilde{N}}[\mathscr{T},\{\tilde{w}^m_{n}\}]$ in the PM scenario. $\tilde{M}=M\cup T$ and $\tilde{N}=T\times N$, respectively, denote the input and output set of Alice and Bob. $\mathscr{T}: T\to T$ is the resource noiseless classical channel, {i.e.}, $\mathscr{T}(\tau'|\tau)=\delta_{\tau',\tau}$. The payoff for this task is given as
\begin{align}\label{eq:app_proofwirereading3}
    \Tilde{S}(\Bar{\mathcal{N}})=\sum_{m\in \tilde{M},(t,n_0)\in \tilde{N}} \tilde{w}^m_{t,n_0} ~ \Bar{\mathcal{N}}(t,n_0|m)\nonumber\\
    \text{ where } \tilde{w}^m_{t,n_0}=\begin{cases}
        &w^m_{t,n_0} \quad\quad\quad\quad\quad \text{ for } m\in M\\
        &-\Theta(1-\delta_{t,m}) \quad\text{ for } m\in \Tilde{M}\backslash M
    \end{cases}
\end{align}

Here, $-\Theta$ is a large penalty and can be considered to be $-\infty$\footnote{Similar penalty has been used in \cite{Alimuddin2023}, where Bob must avoid certain outputs conditioned on the input of Alice.}, necessitating zero-error guessing of Alice's input by Bob when $m\in T$. $\Bar{\mathcal{N}}(t,n_0|m)$ denote the conditional probability of Bob giving output $(t,n_0)\in \tilde{N}=T\times N$ when Alice gets input $m\in \tilde{M}$. In this case, consider the assistance from the shared correlation $P$ as before. Say Alice follows the strategy specified by function $\Tilde{\mathbb{FX}}: \tilde{M}\to [X]$ for choosing the query $x\in [X]$ to the correlation $P$. Here, $\Tilde{\mathbb{FX}}(m)=\mathbb{FX}(m)$ for $m\in M$ and $\Tilde{\mathbb{FX}}(m)=1$ otherwise. Alice uses the function $\Tilde{\mathbb{FT}}: \tilde{M}\times [A]\to T$ to send $\tau\in T$ using the channel where $\Tilde{\mathbb{FT}}(m,a)=\mathbb{FT}(m,a)$ for $m\in M$ and $\Tilde{\mathbb{FT}}(m,a)=m$ otherwise. Bob uses $\mathbb{FY}:T\to [Y]$ to decide the query to $P$ based on the received message and $\mathbb{FN}_1:T\times [B]\to T\times N$ where $\mathbb{FN}_1:=id_T\times\mathbb{FN}$, {\it i.e.}, $\mathbb{FN}_1 (\tilde{t},b)=(\tilde{t},\mathbb{FN}(\tilde{t},b))$ to output $(t,n_0)\in \tilde{N}$. In this case, the simulated conditional probability $\Bar{\mathcal{N}_P}(t,n_0|m)$ is

\begin{align}\label{eq:app_proofwirereading4}
\Bar{\mathcal{N}_P}~(t, n_0|m) &= \sum_{\substack{a, b, x, y,\tau,\tau'}}  \big[ \delta_{x,\Tilde{\mathbb{FX}}(m)}~\delta_{\tau,\Tilde{\mathbb{FT}}(m,a)}~  P(a,b|x,y)~ \mathscr{T}(\tau'|\tau) \delta_{y,\mathbb{FY}(\tau')}~\delta_{(t,n_0),\mathbb{FN}_1(\tau',b)}\big]\nonumber\\
&=\sum_{\substack{a, b, x, y, \tau'}}  \big[ \delta_{x,\Tilde{\mathbb{FX}}(m)}~\delta_{\tau',\Tilde{\mathbb{FT}}(m,a)}~  P(a,b|x,y)~ \delta_{y,\mathbb{FY}(\tau')}~\delta_{(t,\tau')}~\delta_{n_0,\mathbb{FN}(\tau',b)}\big]\nonumber\\
&=\sum_{\substack{a, b, x, y}}  \big[ \delta_{x,\Tilde{\mathbb{FX}}(m)}~\delta_{t,\Tilde{\mathbb{FT}}(m,a)}~  P(a,b|x,y)~ \delta_{y,\mathbb{FY}(t)}~\delta_{n_0,\mathbb{FN}(t,b)}\big]
\end{align} 

Clearly,  $\Bar{\mathcal{N}_P}~(t, n_0|m)=\mathcal{N}_P~(t, n_0|m)$ for $m\in M$ and $\sum_{n_0}\Bar{\mathcal{N}_P}~(t= m, n_0|m)=1$ for $m\in T$. $\Tilde{S}(\Bar{\mathcal{N}_P})=\sum_{m\in \tilde{M},(t,n)\in \tilde{N}} \tilde{w}^m_{t,n}~\Bar{\mathcal{N}_P}(t,n|m)=\sum_{m\in M,(t,n)\in \tilde{N}} \tilde{w}^m_{t,n}~\Bar{\mathcal{N}_P}(t,n|m)=S_{W}^{opt}$. Thus, using the aforementioned strategy and the shared correlation $P$ leads to the same payoff $S_W^{opt}$ in the task $\mathbb{CT}_{\tilde{M},\tilde{N}}[\mathscr{T},\{\tilde{w}^m_{n}\}]$. If the optimal payoff for the task $\mathbb{CT}_{\tilde{M},\tilde{N}}[\mathscr{T},\{\tilde{w}^m_{n}\}]$ while using assistance from correlation is denoted by $S^{opt}$, then $S_{W}^{opt}\leq S^{opt}$. Next, we will show that $S^{opt}\leq S_{W}^{opt}$.\\

Since the success metric in \eqref{eq:app_proofwirereading3} is linear in $\Bar{\mathcal{N}}(t,n_0|m)$, the payoff $S^{opt}$ when using assistance from correlation $P$ can be achieved while using some deterministic strategy as following. Alice uses the function $\Bar{\mathbb{FX}}: \tilde{M}\to [X]$ for choosing the query $x\in [X]$ to the correlation $P$. She uses the function $\Bar{\mathbb{FT}}: \tilde{M}\times [A]\to T$ to send $\tau\in T$ using the channel based on input $m$ and output from the correlation $a$. Bob uses the strategy specified by $\Bar{\mathbb{FY}}:T\to [Y]$ to query $y= \Bar{\mathbb{FY}}(\tau')$ to the correlation upon receiving message $\tau'$. $\Bar{\mathbb{FN}}_1:T\times [B]\to T\times N$ is used to output $(t,n_0)=\Bar{\mathbb{FN}}_1(\tau',b)\in \tilde{N}$ upon obtaining $b$ from correlation as output. In this case, the simulated conditional probability $\Bar{\mathcal{N}'_P}(t,n_0|m)$ is

\begin{align}\label{eq:app_proofwirereading5}
\Bar{\mathcal{N}'_P}~(t, n_0|m) &= \sum_{\substack{a, b,x, y, \tau'}}  \big[ \delta_{x,\Tilde{\mathbb{FX}}(m)}~\delta_{\tau',\Bar{\mathbb{FT}}(m,a)}~  P(a,b|x,y)~ \delta_{y,\Bar{\mathbb{FY}}(\tau')}~\delta_{(t,n_0),\Bar{\mathbb{FN}}_1(\tau',b)}\big]
\end{align} 

For the case when Alice's input $m\in T$, Bob's output $t\in T$ must be the same $m$ to avoid the penalty defined in the payoff function \eqref{eq:app_proofwirereading3}. For $\sum_{n_0}\Bar{\mathcal{N}'_P}~(t=m, n_0|m)=1~  \forall~ m\in T$, $\tau=\tau'= \Bar{\mathbb{FT}}(m)=Perm(m)~ \forall~ m\in T$ and $(t,n_0)=\Bar{\mathbb{FN}}_1(\tau',b)=(Perm^{-1}(\tau'),\Bar{\mathbb{FN}}(\tau',b))$ where $Perm:T\to T$ is a permutation and $\Bar{\mathbb{FN}}:T\times [B]\to N$ is some function that is used to decide Bob's other output $n_0$. Substituting this in \eqref{eq:app_proofwirereading5}, for $m\in M$, we get
\begin{align}\label{eq:app_proofwirereading6}
\Bar{\mathcal{N}'_P}~(t, n_0|m) &= \sum_{\substack{a, b, x, y, \tau'}}  \big[ \delta_{x,\Bar{\mathbb{FX}}(m)}~\delta_{\tau',\Bar{\mathbb{FT}}(m,a)}~  P(a,b|x,y)~ \delta_{y,\Bar{\mathbb{FY}}(\tau')}~\delta_{(t,n_0),(Perm^{-1}(\tau'),\Bar{\mathbb{FN}}(\tau',b))}\big]\nonumber\\
&=\sum_{\substack{a, b, x, y, \tau'}}  \big[ \delta_{x,\Bar{\mathbb{FX}}(m)}~\delta_{\tau',\Bar{\mathbb{FT}}(m,a)}~  P(a,b|x,y)~ \delta_{y,\Bar{\mathbb{FY}}(\tau')}~\delta_{t,Perm^{-1}(\tau')}\delta_{n_0,\Bar{\mathbb{FN}}(\tau',b)}\big]\nonumber\\
&=\sum_{\substack{a, b, x, y}}  \big[ \delta_{x,\Bar{\mathbb{FX}}(m)}~\delta_{Perm(t),\Bar{\mathbb{FT}}(m,a)}~  P(a,b|x,y)~ \delta_{y,\Bar{\mathbb{FY}}(Perm(t))}~\delta_{n_0,\Bar{\mathbb{FN}}(Perm(t),b)}\big]\nonumber\\
&=\sum_{\substack{a, b, x, y}}  \big[ \delta_{x,\Bar{\mathbb{FX}}(m)}~\delta_{t,\Bar{\mathbb{FT}}'(m,a)}~  P(a,b|x,y)~ \delta_{y,\Bar{\mathbb{FY}}'(t)}~\delta_{n_0,\Bar{\mathbb{FN}}'(t,b)}\big]
\end{align}
where $\Bar{\mathbb{FY}}'=\Bar{\mathbb{FY}}\circ Perm$ and $\Bar{\mathbb{FN}}'(\tau,b)=\Bar{\mathbb{FN}}(Perm(\tau),b)$. Also, $\Bar{\mathbb{FT}}'=Prem^{-1}\circ \Bar{\mathbb{FT}}$. Note that the strategy leads to the payoff $S^{opt}$.\\

In the task $\mathbb{CW}_{M,N}[\mathscr{T},\{w^m_{\tau',n}\}]$, say Alice uses the strategy specified by function $\mathbb{FX}: M\to [X]$ for choosing the query to the correlation $P$ where $\mathbb{FX}(m)=\Bar{\mathbb{FX}}(m)$. She uses the function $\mathbb{FT}: M\times [A]\to T$ to send $\tau\in T$ using the channel where $\mathbb{FT}(m,a)=\Bar{\mathbb{FT}}'(m,a)$ for $m\in M,a\in [A]$. Bob uses the strategy specified by $\mathbb{FY}:T\to [Y]$ where $\mathbb{FY}(\tau')=\Bar{\mathbb{FY}}(\tau')$ for choosing query to $P$ based on received message $\tau'$ and decoding function $\Bar{\mathbb{FN}}':T\times[B]\to N$. The conditional probability distribution $\{\mathcal{N}'_P~(\tau', n|m)\}_{m,\tau',n}$ they obtain has same expression as in \eqref{eq:app_proofwirereading6} for $m\in M$. As $\tilde{w}^m_{t,n}=w^m_{t,n}$ for $m\in M$, the payoff in this case is also $S^{opt}$. Therefore, for the correlation $P$, the optimal payoff in both tasks is equal, i.e., $S_{W}^{opt}=S^{opt}$.
\end{proof}

\section{Proof of Theorem \ref{theo: CS0 classical}}\label{app:proof CS0 classical}
\begin{proof}
For the task, say Bob uses some deterministic decoding given by function $\mathbb{D}: [d]\to N$ to output $n=\mathbb{D}(\tau')\in N$ when $\tau'$ is the received message. Consider a particular input $m=\{m_1,\cdots m_d\}\in M$ for Alice where $m_i=\mathbb{D}(i)$ for $i\in\{1,\cdots,d\}$. If Alice follows a deterministic encoding $\mathbb{E}: M\to [d]$, she sends $\tau= \mathbb{E}(m)$ for this particular input. Bob's output on receiving message $\tau'=\tau$, is determined as $n=\mathbb{D}(\tau')=m_{\tau'}$. The conditional probability $\mathcal{N}(\tau'=\mathbb{E}(m),n= \mathbb{D}(\tau')|m)=\mathcal{N}(\tau'=\mathbb{E}(m),n=m_{\tau'}|m)=1\implies \mathcal{N}(\tau',n\neq m_{\tau'}|m)=0$ for this input $m$. Thus, the payoff corresponding to this specific input $m$ is $0$ following the encoding $\mathbb{E}$ and decoding $\mathbb{D}$. Note that the deterministic encoding by Alice $\mathbb{E}$ is not specified by Bob's choice of decoding $\mathbb{D}$. Also, for every choice of deterministic decoding $\mathbb{D}: [d]\to N$ by Bob we can find an input of Alice $m\in M$ such that payoff corresponding to this specific input $0$ when Alice uses any arbitrary deterministic encoding $\mathbb{E}: M\to [d]$. Thus, the payoff is upper bounded by  $(S_{W}^{alg}-\frac{1}{k^d})=(1-\frac{1}{k^d})$ for any arbitrary deterministic encoding and decoding strategies by Alice and Bob. Using the linearity of the payoff function in \eqref{eq: CS1 payoff} implies that the maximum payoff can be achieved while using some deterministic strategy. Thus, the shared randomness assisted bound for the task $\mathbb{CS}[d,k]$ is $s_{\Lambda}\leq (1-\frac{1}{k^d})$.\\

    We will show a shared randomness-assisted strategy that achieves this upper bound. Say, Bob uses the deterministic decoding $\mathbb{D}: [d]\to N$ where $\mathbb{D}(\tau')=1 ~\forall \tau'\in [d]$. Alice uses the encoding $\mathbb{E}: M\to [d]$ as follows.  If $m=\{1,1,\cdots,1\}$, $\tau=E(m)=1$ otherwise $\tau=E(m)=\min\{i\in [d]|m_i>1\}$ for input $m=(m_1,\cdots,m_d)\in M$. Bob's output decoding the received message $\tau'=\tau$ is $n=1\ne m_\tau'$ only if $m\neq\{1,1,\cdots,1\}$. The payoff in this case is $\sum_{m\in M, m\neq \{1,1,\cdots,1\}} (\frac{1}{k^d})=(1-\frac{1}{k^d})$.
\end{proof}

\section{Proof of Lemma \ref{lemma:nlfacet correlation property}}\label{app:lemma nlfacet correlation}
\begin{proof}
    For the non-signalling polytope $\mathcal{NS}$, let $\{P_{L}^{(i)}\}$ and $\{P_{NL}^{*(i)}\}$ represent the sets of local and non-local extremal points, respectively. Every local extremal point $P_{L}^{(i)}$ can be described by a pair of functions $\mathbb{LA}:[X]\to [A]$ and $\mathbb{LB}:[Y]\to [B]$. A correlation $P_{NL}=\{P_{NL}(a,b|x,y)\}_{a\in [A],b\in [B],x\in [X],y\in [Y]}$ belongs to a non-local facet of $\mathcal{NS}$ polytope if either $P_{NL}\in\{P_{NL}^{*(i)}\}$ or it can be \emph{only} expressed as some convex mixture of its elements. Now, for some arbitrary function $\mathbb{LB}:[Y]\to [B]$, let us assume that the statement of the Lemma is false. Then, for all $x\in [X]$ there exists at least one $a_x\in [A]$ such that $\Phi^{x,a_x}_{\mathbb{LB}}(P_{NL})= \varnothing$. In other words, $\forall~ y\in [Y],$ $ P_{NL}(a_x,b=\mathbb{LB}(y)|x,y)>0\implies P_{NL}(b=\mathbb{LB}(y)|x,ya_x)>0\land P_{NL}(a_x|x)>0$ using the Bayes rule. Using such outcomes for each $x\in [X]$, we can define a function $\mathbb{LA}:[X]\to [A]$ where $\mathbb{LA}(x)=a_x\in[A]$ such that $P_{NL}(a_x|x)>0$ and $\Phi^{x,a_x}_{\mathbb{LB}}(P_{NL})= \varnothing$.\\

    Now consider the local extremal point $P_{L}=\{P_{L}(a,b|x,y)\}_{a\in [A],b\in [B],x\in [X],y\in [Y]}$ corresponding to $\mathbb{LA}:[X]\to [A]$ and $\mathbb{LB}:[Y]\to [B]$ pair, {\it i.e.}, $P_{L}(a,b|x,y)=\delta_{a,\mathbb{LA}(x)}\delta_{b,\mathbb{LB}(y)}$. Clearly, $\forall x,y~ P_{L}(a,b|x,y)=1 \implies P_{NL}(a,b|x,y)>0$. Consequently we can decompose $P_{NL}$ as following  $P_{NL}=pP_{L}+(1-p)\tilde{P}$ where $0<p\leq 1$ and $\tilde{P}\in \mathcal{NS}$. However, this is a contradiction as $P_{NL}$ is a correlation on the non-local facet of $ \mathcal{NS}$ polytope and does not admit a convex decomposition with a positive weight corresponding to any local extremal correlation $P_{L}\in\{P_{L}^{(i)}\}$.
\end{proof}

\section{Proof of Theorem \ref{theo: CS0 ns facet}}\label{app:CS0 nlfacet correlation}
\begin{proof}
    In this communication task, say the communication channel $\mathscr{T}$ is assisted by some correlation $P_{NL}=\{P_{NL}(a,b|x,y)\}_{a\in [A],b\in [k],x\in [X],y\in [d]}$ from a non-local facet of the no-signalling polytope $\mathcal{NS}:=\{P=\{P(a,b|x,y)\}_{a\in [A],b\in [k],x\in [X],y\in [d]}\}$. Corresponding to each input $m=(m_1,\cdots,m_d)\in M=\{1,\cdots,k\}^d$, we can consider a function $\mathbb{LB}:[d]\to [k]$ such that $\mathbb{LB}(i)=m_i \forall~ i\in [d]$. Using Lemma \ref{lemma:nlfacet correlation property}, for the function $\mathbb{LB}:[d]\to [k]$ corresponding to $m$ and the correlation $P_{NL}$, there exists an $x_m\in [X]$ such that for all $a\in [A]$ there exists a non-empty subset $\Phi^{x_m,a}_{\mathbb{LB}}(P_{NL}) \subseteq [d]$. For all $y\in \Phi^{x_m,a}_{\mathbb{LB}}(P_{NL})$,   $P_{NL}(a,b=\mathbb{LB}(y)|x_m,y)=0$. Based on this, the protocol Alice and Bob use for the input $m$ is as follows:
    \begin{itemize}
        \item {\bf Encoding:} For the input $m$ Alice queries $x_m\in[X]$ to the correlation $P_{NL}$. Obtaining output $a\in [A]$, she sends some particular $y_{(m,a)}\in \Phi^{x_m,a}_{\mathbb{LB}}(P_{NL})\subseteq [d]$ as message $\tau$.
        \item {\bf Decoding:} Bob on receiving message $\tau'=\tau$ uses it as query, {\it i.e.} $y=\tau'=y_{(m,a)}$, to the correlation $P_{NL}$ and gives the outcome $b\in[k]$ as the final output, {\it i.e.}, $n=b$.
    \end{itemize}
Following the above encoding and decoding while sharing correlation $P_{NL}$, 
\begin{align}\label{eq: CS0 nsfacet corr}
\mathcal{N}(\tau',n|m)&=\sum_{\substack{x,a,y,b,\tau}}\substack{[\delta_{x,x_m}\delta_{\tau,y_{(m,a)}} \delta_{\tau,\tau'}\delta_{y,\tau'} P_{NL}(a,b|x,y)\delta_{n,b}]}= \sum_{\substack{a}} \substack{[\delta_{\tau',y_{(m,a)}} P_{NL}(a,n|x_m,y_{(m,a)})]}
\end{align}
For the input $m$, $\mathcal{N}(\tau',n|m)=0$ if $\tau'\notin \{y_{(m,a)}\}_{a\in [A]}$ from the protocol. Also, $\mathcal{N}(\tau',n|m)=0$ if $n=m_{\tau'}$ since $\sum_{\substack{a}}\delta_{\tau',y_{(m,a)}} P_{NL}(a,n|x_m,y_{(m,a)})=0$ in this case from the choice of $y_{(m,a)}$ for each $a\in [A]$. Thus, $\mathcal{N}(\tau',n\neq m_{\tau'}|m)=1$. Thus, the payoff while using the above encoding and decoding is $\sum_{m\in M}(\frac{1}{k^d})=1$.
\end{proof}

\section{Proof of Corollary \ref{corr: CS0 noisy nl facet}}\label{app:CS0 noisy nlfacet correlation}
\begin{proof}
Let Alice and Bob share the correlation  correlation $P_{NL}=\{P_{NL}(a,b|x,y)\}_{a\in [A],b\in [k],x\in [X],y\in [d]}$ on a non-local facet of the no-signalling polytope $\mathcal{NS}:=\{P=\{P(a,b|x,y)\}_{a\in [A],b\in [k],x\in [X],y\in [d]}\}$. Then they can follow the protocol as specified in the proof of Theorem \ref{theo: CS0 ns facet} and obtain a payoff of $S_{W}^{alg}=1$ for the task. Say, they also follow the same protocol while sharing the correlation $P_{WN}=\{P_{WN}(a,b|x,y)=\frac{1}{Ak}\}_{a\in [A],b\in [k],x\in [X],y\in [d]}$. For input $m\in M$, if Alice choses $x_m\in [X]$ as the query to the correlation $P_{WN}$ (as before) and obtaining output $a\in [A]$ sends $\tau=y_{(m,a)}\in \Phi^{x_m,a}_{\mathbb{LB}}(P_{NL})\subseteq[d]$. For the input $m$, let us denote by $r_{\tau,m}$ the multiplicity of $a\in[A]$ such that Alice sends a particular $\tau\in [d]$ as the message. So, $\sum_{\tau\in[d]}r_{\tau,m}=A~\forall ~ m\in M$. Using the correlation $P_{WN}$ and the same protocol as before
\begin{align}\label{eq: CS0 whitenoisecorr}
\mathcal{N}(\tau',n|m)&= \sum_{\substack{a\in[A]}} [\delta_{\tau',y_{(m,a)}}p_{WN}(n|a,x_m,y_{(m,a)})p_{WN}(a|x_m)]=\sum_{\substack{a\in [A
]}} [\delta_{\tau',y_{(m,a)}}\frac{1}{k} \frac{1}{A}]=\frac{1}{kA}r_{\tau',m}
\end{align}
Note that, for $m\in M$ and $\tau'\in [d]$, there is exactly one $n\in [k]$ such that $n=m_{\tau'}$. Thus, $\mathcal{N}(\tau',n\neq m_{\tau'}|m)=\frac{1}{kA}r_{\tau',m}(k-1)$. Thus, the payoff following the above strategy while using assistance from $P_{WN}$ is 
\begin{align}\label{eq: CS0 whitenoisepayoff}
    S_{W}(\mathcal{N})&=\sum_{m\in M}[\sum_{\tau'\in[d]}\frac{1}{k^d}\mathcal{N}(\tau',n\neq m_{\tau'}|m)]\nonumber\\
    &=\sum_{m\in M}[\sum_{\tau'\in[d]}\frac{1}{k^d}\frac{1}{kA}r_{\tau',m}(k-1)]\nonumber\\
    &=\sum_{m\in M}\frac{1}{k^{d+1}A}(k-1)[\sum_{\tau'\in[d]}r_{\tau',m}=\sum_{m\in M}\frac{1}{k^{d+1}}(k-1)=\frac{(k-1)}{k}
\end{align}
Since the payoff function is linear in the conditional probabilities, using the same strategy while sharing the correlation $pP_{NL}+(1-p)P_{WN}$ where $0\leq p\leq1$, the payoff is $p (1)+(1-p)\frac{(k-1)}{k}$. Now, the correlation is advantageous over shared randomness assistance if $p (1)+(1-p)\frac{(k-1)}{k}>1-\frac{1}{k^d}\implies p> (1-\frac{1}{k^{d-1}})$.
\end{proof}

\section{Proof of Theorem \ref{theo: CS1 classical bound}}\label{app:CS1 classical bound}
In the task $\mathbb{CS}[\{P_{NL}^{*(i)}\}_{i\in I},d,k]$ with the payoff function as defined in \eqref{eq: CS2 payoff}, when Alice and Bob have access to shared randomness as assistance to the classical communication then the maximum payoff can always be achieved if Alice and Bob follow some optimal deterministic encoding and decoding scheme as their strategy. Say, Bob uses some deterministic decoding given by function $\mathbb{D}: [d]\to [k]$ to output $n=\mathbb{D}(\tau')\in [k]$ when $\tau'$ is the received message. In this case, consider the inputs of Alice, $m\in M$, which is a relation such that for all $t\in [d],~ (t,\mathbb{D}(t))\notin m$. Let us denote the set of such inputs by $M_{\mathbb{D}}$. Example of such inputs include $m\in\mathcal{R}_{\mathbb{LB}}$ where the function $\mathbb{LB}(t)=\mathbb{D}(t)~ \forall t\in [d]$. Say Alice uses an arbitrary deterministic encoding strategy $\mathbb{E}: M\to [d]$, where she sends the message $\tau= \mathbb{E}(m)$ corresponding to the input $m\in M$. The observed conditional probability $\mathcal{N}(\tau',n|m)=0$ if $m\in M_{\mathbb{D}}$ and $(\tau',n)\in m$, using the definition of set $M_{\mathbb{D}}$. Thus, the payoff using a deterministic encoding strategy $\mathbb{E}$ and decoding strategy $\mathbb{D}$ leads to a payoff $0$ for inputs $m\in M_{\mathbb{D}}$. The payoff using this strategy is  $\leq \sum_{m\in M\backslash M_{\mathbb{D}}} \omega^m$. When Bob uses the decoding $\mathbb{D}$, then the payoff $ \sum_{m\in M\backslash M_{\mathbb{D}}} \omega^m$ can be achieved if Alice uses an optimal deterministic encoding as follows. For each $m\in M\backslash M_{\mathbb{D}}$ she chooses to send $\tau \in [d]$ such that $(\tau'=\tau,\mathbb{D}(\tau'))\in m$. The existence of such an $\tau\in[d]$ for each $m\in M\backslash M_{\mathbb{D}}$ follows directly from the definition of the set $M_{\mathbb{D}}$.\\

Thus, maximum payoff when Bob uses the deterministic decoding strategy $\mathbb{D}: [d]\to N$ is $ \sum_{m\in M\backslash M_{\mathbb{D}}} \omega^m$. Now, optimising over all deterministic decoding by Bob gives the maximum payoff which can be obtained using shared randomness assistance to the classical channel, {\it i.e.}, 

\begin{align}
    s_{\Lambda}=\max_{\mathbb{D}:[d]\to [k]}\left[\sum_{\substack{m\in M\backslash M_{\mathbb{D}}}} \omega_m\right]=\max_{\mathbb{D}:[d]\to [k]}\left[\sum_{\substack{m\in M: \\ \exists \alpha\in[d]\land(\alpha,\mathbb{D}(\alpha))\in m}} \omega_m\right]
\end{align}
Note that for each decoding strategy  $\mathbb{D}:[d]\to [k]$, there exists input $m\in \mathcal{R}_{\mathbb{LB}(=\mathbb{D})}$ such that $\omega_m> 0$ (from the description of the task). As described before, any deterministic encoding choice of Alice leads to a $0$ payoff for the input if $m\in M_{\mathbb{D}}$. Thus, $s_{\Lambda}< S_{W}^{alg}=\sum_{m\in M}\omega_m=1$.

\section{Proof of Theorem \ref{theo: CS1 nl facet}}\label{app:CS1 nl facet}

For $\mathbb{CS}[\{P_{NL}^{*(i)}\}_{i\in I},d,k]$, say Alice and Bob share some correlation $P_{NL}\in Face\{P_{NL}^{*(i)}\}$ on the non-local facet of the no-signalling polytope $\mathcal{NS}:=\{P=\{P(a,b|x,y)\}_{a\in [A],b\in [k],x\in [X],y\in [d]}\}$ and use it for assistance to classical communication. Note that from the task description in Section \ref{subsec: type1CS}, each input of Alice $m\in M$ is some relation which can be described in the following way. Corresponding to the relation $m \subset [d]\times [k]$ there exists a function $\mathbb{LB}:[d]\to [k]$ and  $x^*\in [X]$, independent of the choice of $P_{NL}^{*(i)}$ where $i\in I$ such that: for all $a\in [A]$ if $\cup_{y\in\Phi^{x^*,a}_{\mathbb{LB}}}~\Phi^{x^*,a,y}_{\mathbb{LB}}\neq \varnothing$, then there exists $y\in \Phi^{x^*,a}_{\mathbb{LB}}$, such that $(y,b)\in m$ for all $b\in \Phi^{x^*,a,y}_{\mathbb{LB}}$.\\

Say Alice and Bob use the following encoding and decoding strategy. Alice, receiving the input $m\in M$, identifies the corresponding function $\mathbb{LB}:[d]\to [k]$ and uses the $x^*\in [X]$ as the query to the shared correlation $P_{NL}$. Upon obtaining the outcome $a\in [A]$ if $\cup_{y\in\Phi^{x^*,a}_{\mathbb{LB}}}~\Phi^{x^*,a,y}_{\mathbb{LB}}\neq \varnothing$, she sends the message $\tau=y\in \Phi^{x^*,a}_{\mathbb{LB}}$ such that $(y,b)\in m$ for all $b\in \Phi^{x^*,a,y}_{\mathbb{LB}}$. The existence of such $y$ is guaranteed by condition (ii) that the relation $m$ satisfies. For outcomes $a\in [A]$ such that $\cup_{y\in\Phi^{x^*,a}_{\mathbb{LB}}}~\Phi^{x^*,a,y}_{\mathbb{LB}}=\varnothing$, she send message $\tau=1$. Bob, on receiving message $\tau'=\tau$, uses it as his query to the correlation $P_{NL}$ and on obtaining outcome $b\in [k]$ gives it as his output, {\it i.e.}, $n=b$.\\

If $P_{NL}=P_{NL}^{*(i)}$ where $i\in I$ . For the input $m\in M$, on using $x^*$ as query,  if $P_{NL}^{*(i)}(a|x^*)>0$ for some $a\in [A]$ then $\cup_{y\in\Phi^{x^*,a}_{\mathbb{LB}}}~\Phi^{x^*,a,y}_{\mathbb{LB}}\neq \varnothing$. Corresponding to each such output $a\in [A]$, Bob's query upon receiving Alice's message $\tau'=\tau\in \Phi^{x^*,a}_{\mathbb{LB}}$ implies that the outcome $b'\in [k]$, {\it i.e.}, $P_{NL}^{*(i)}(b'|x^*,\tau',a)>0$, also belongs to $\Phi^{x^*,a,\tau'}_{\mathbb{LB}}$. From the protocol and using the fact that $m$ satisfies condition (ii), it implies that always $(\tau',n=b')\in m$ for the input $m\in M$ leading to a payoff $\omega_m$. Thus, the strategy described for all $m\in M$ while sharing correlation $P_{NL}^{*(i)}$ guarantees payoff $\sum_{m\in M}\omega_m=1=S_{W}^{alg}$.\\

As the strategy is invariant over shared correlation $P_{NL}=P_{NL}^{*(i)}$, where $i\in I$, and leads to payoff $S_{W}^{alg}$, therefore it will also give a payoff $S_{W}^{alg}$ when $P_{NL}\in Face\{P_{NL}^{*(i)}\}$ using the linearity of the payoff in the conditional probability distribution $\{ \mathcal{N}(\tau',n|m)\}_{\tau',n,m}$.

\section{Proof of part (i) of Corollary \ref{corr:CS1 xy22 srbound}}\label{app:CS1 classical bound xy22}
In task $\mathbb{CS}[\{P_{NL}^{*}\},d,k=2]$, when Alice and Bob have access to shared randomness as assistance to the classical communication then the maximum payoff can always be achieved if Alice and Bob follow some optimal deterministic encoding and decoding scheme as their strategy. Say, Bob uses some deterministic decoding given by function $\mathbb{D}: [d]\to [2]$ to output $n=\mathbb{D}(\tau')\in [2]$ when $\tau'$ is the received message. Consider Alice's input $m\in M$ where $m=\mathrm{R}_{ijkl}$ such that $k \neq\mathbb{D}(i)$ and $l \neq\mathbb{D}(j)$. Let us denote the set of such inputs by $M_{\mathbb{D}}$. Since one such input is possible for each choice of $i,j$, thus $|M_{\mathbb{D}}|={d\choose 2}$. If Alice uses an arbitrary deterministic encoding strategy $\mathbb{E}: M\to [d]$, where she sends the message $\tau= \mathbb{E}(m)$ corresponding to the input $m\in M$. The observed conditional probability $\mathcal{N}(\tau',n|m)=0$ if $m\in M_{\mathbb{D}}$ and $(\tau',n)\in m$, using the definition of set $M_{\mathbb{D}}$. Thus, the payoff using a deterministic encoding strategy $\mathbb{E}$ and decoding strategy $\mathbb{D}$ leads to a payoff $0$ for inputs $m\in M_{\mathbb{D}}$. The payoff using this strategy is  $\leq \sum_{m\in M\backslash M_{\mathbb{D}}} \omega^m=[4{d\choose 2}-{d\choose 2}]\frac{1}{2d(d-1)}=\frac{3}{4}$.  \\

This upper bound is independent of the choice of the deterministic decoding strategy $\mathbb{D}$ and can be achieved by using Alice's optimal encoding. For input $m\in M\backslash M_{\mathbb{D}}$, $m=\mathrm{R}_{ijkl}$ such that either $k =\mathbb{D}(i)$ or $l=\mathbb{D}(j)$. Alice sends the message $\tau= \mathbb{E}(m)=i$ if $k =\mathbb{D}(i)$ else she sends $\tau= \mathbb{E}(m)=j$ corresponding to the input $m\in M\backslash M_{\mathbb{D}}$. Bob's decoding $\mathbb{D}$ guarantees a payoff $\omega_m=\frac{1}{2d(d-1)}$ for the input $m\in M\backslash M_{\mathbb{D}}$.

\section{Proof of part (ii) of  Corollary \ref{corr:CS1 xy22 srbound}}\label{app:CS1 nl facet bound xy22}
In $\mathbb{CS}[\{P_{NL}^{*}\},d,k=2]$, say Alice and Bob use shared correlation $P_{NL}^{*}$ to assist the communication channel. We will present a protocol for the parties and show that using it gives a payoff $S_{W}^{alg}=1$. Bob, on receiving message $\tau'=\tau$, uses it as his query to the correlation $P_{NL}^{*}$ and upon obtaining outcome $b\in [2]$, gives it as his output, {\it i.e.}, $n=b$.\\

Alice's input $m\in M$ corresponds to some relation $\mathrm{R}_{ijkl}:=\{(i,k),(j,l)\}\subset [d]\times [2]$ where $i<j$, $i,j\in [d]$ and $k,l\in [2]$.
Now we will provide an encoding strategy for Alice.  First consider inputs $m=\mathrm{R}_{ijkl}$ where $i,j$ is fixed. Corresponding to $y'=i,y''=j$, from the description of the extremal correlation $P_{NL}^{*}$, there is some $x',x''(>x')\in [X]$ for which $\{P_{NL}^{*}(a,b|x,y)\}_{a\in[2],b\in [2],x\in\{x',x''\},y\in\{y',y''\}}$ is equivalent to a two-input two-output PR-box as in equation \eqref{eq:PR boxes}. For simplicity, assume it is equivalent to $\{P_{NL}^{*}(a,b|x,y)\}_{a\in[2],b\in [2],x\in\{x',x''\},y\in\{y',y''\}}=P_{NL}^{*(0,0,0)}$, {\it i.e.}, $\alpha=0,\beta=0,\gamma=0$. Similar encodings for the other assignments of $\alpha,\beta,\gamma$ can be obtained, all of which are functionally equivalent to the one we present. For input relation $m= \mathrm{R}_{ijkl}$, if $k=l$,
Alice uses $x''\in [X]$ as her query to the correlation $P_{NL}^{*}$, and if $k\neq l$, Alice uses $x'\in [X]$ as her query. If $k=1$, depending on her outcome $a\in [2]$, she sends the message $\tau=i$ if $a=1$ and $\tau=j$ if $a=2$. If $k=2$, she sends the message $\tau=j$ if $a=1$ and $\tau=i$ if $a=2$.  For each input $m= \mathrm{R}_{ijkl}$, where $i,j$ is fixed, using this strategy along with Bob's decoding ensures that the communicated message $\tau'$ and Bob's output $n$ are such that $(\tau',n)\in m$. Consequently, it results in a payoff $\omega_m$ for each input $m= \mathrm{R}_{ijkl}$ where $i,j$ is fixed.
Notice that $i,j\in [d]$ here are arbitrary (but fixed) indices. Thus, Alice can use a similar encoding for each pair of $i,j\in [d]$. The payoff using such encodings by Alice and the aforementioned decoding by Bob is $\omega_m$ for each $m\in M$. The overall payoff is $\sum_{m\in M}\omega_m=1$.

\section{Proof of Corollary \ref{corr:CS1xy22 noisy nl facet}}\label{app:CS1 noisy nl facet bound xy22}
In $\mathbb{CS}[\{P_{NL}^{*}\},d,k=2]$, say Alice and Bob use shared correlation $P_{WN}=\{P_{WN}(a,b|x,y)=\frac{1}{4}\}_{a\in[2],b[2],x\in [X],y\in [d]}$ to assist the communication channel. Say Alice and Bob use the same encoding and decoding as discussed in the proof of part (ii) of Corollary \ref{corr:CS1 xy22 srbound} (see appendix \ref{app:CS1 nl facet bound xy22}). Clearly for each input $m= \mathrm{R}_{ijkl}$, $\mathcal{N}(\tau',n|m)=\frac{1}{4}$ if $\tau'\in\{i,j\}$ and $n\in [2]$ using the strategy. Thus, for each $m$ as described above the payoff is $\omega_m(\mathcal{N}(i,k|m)+\mathcal{N}(j,l|m))=\frac{\omega_m}{2}$. Thus, the net payoff while using this strategy is $\sum_{m\in M}\frac{\omega_m}{2}=\frac{1}{2}$.\\

Given that the shared correlation is $ P = p\,P_{NL}^{*} + (1 - p)\,P_{WN} $, and the parties employ the strategy outlined in the proof of part (ii) Corollary \ref{corr:CS1 xy22 srbound}, the resulting total payoff is the convex sum of the individual payoffs corresponding to \( P_{NL}^{*} \) and \( P_{WN} \), weighted by \( p \) and \( 1 - p \), respectively. Thus, the payoff is $p(1)+(1-p)\frac{1}{2}$. This payoff is higher than $s_{\Lambda}=\frac{3}{4}$ if $p+\frac{1-p}{2}>\frac{3}{4}\implies p>\frac{1}{2}$.

\section{Proof of Corollary \ref{corr:13322 nl facet}}\label{app:CS1 nl facet bound I322}
In $\mathbb{CS}[\{P_{NL}^{*(1)},P_{NL}^{*(2)}\},d=3,k=2]$, say the parties use correlation $P\in Face\{P_{NL}^{*(1)},P_{NL}^{*(2)}\}$ as assistance to the communication channel. The following protocol for the parties gives a payoff $S_{W}^{alg}=1$. Upon receiving message $\tau'=\tau$, Bob uses it as his query to the correlation $P$ and after obtaining outcome $b\in [2]$, gives it as his output, {\it i.e.}, $n=b$.\\

For inputs $m\in M$ such that $m= \mathrm{R}_{12kl}$, if $k=l$, Alice uses $x=3$ as her query to the correlation $P$, and if $k\neq l$, Alice uses $x=1$ as her query. If $k=1$, depending on her outcome $a\in [2]$, she sends the message $\tau=2$ if $a=1$ and $\tau=1$ if $a=2$.  If $k=2$, depending on her outcome $a\in [2]$, she sends the message $\tau=a$. For other inputs $m\in M$ such that $m= \mathrm{R}_{ijkl}$, if $k=l$, Alice uses $x=1$ as her query to the correlation $P$, and if $k\neq l$, Alice uses $x=2$ as her query. For inputs these inputs $m= \mathrm{R}_{ijkl}$ if $k=1$, depending on her outcome $a\in [2]$, she sends the message $\tau=i$ if $a=1$ and $\tau=j$ if $a=2$.  If $k=2$, depending on her outcome $a\in [2]$, she sends the message $\tau=j$ if $a=1$ and $\tau=i$ if $a=2$. For each input $m= \mathrm{R}_{ijkl}$, using the above protocol ensures that the communicated message $\tau'$ and Bob's output $n$ are such that $(\tau',n)\in m$. Consequently, it results in a payoff $\omega_m$ for each input $m$. The overall payoff is therefore $\sum_{m\in M}\omega_m=1$.
\end{document}